\documentclass[oneside,english]{amsart}
\usepackage[T1]{fontenc}
\usepackage[latin9]{inputenc}
\usepackage{color}
\usepackage{amssymb}

\makeatletter
\numberwithin{equation}{section} %% Comment out for sequentially-numbered
\numberwithin{figure}{section} %% Comment out for sequentially-numbered
  \theoremstyle{plain}
  \newtheorem*{thm*}{Theorem}
  \theoremstyle{plain}
  \newtheorem{thm}{Theorem}[section]
  \theoremstyle{definition}
  \newtheorem{defn}[thm]{Definition}
  \theoremstyle{plain}
  \newtheorem{lem}[thm]{Lemma}
  \theoremstyle{plain}
  \newtheorem{prop}[thm]{Proposition}
  \theoremstyle{remark}
  \newtheorem{rem}[thm]{Remark}
  \theoremstyle{remark}
  \newtheorem{ex}[thm]{Example}
  \theoremstyle{remark}
  \newtheorem*{acknowledgement*}{Acknowledgement}

\@addtoreset{equation}{section}

\usepackage{babel}
\makeatother

\begin{document}

\title[Symmetries, Newtonoid vector fields and conservation laws in
  the Lagrangian $k$-symplectic formalism] {Symmetries, Newtonoid vector fields and conservation laws in the Lagrangian $k$-symplectic formalism}

\author[Bua]{Luc\'{\i}a Bua}
\address{Luc\'{\i}a Bua, Departamento de Xeometr\'{\i}a e Topolox\'{\i}a, Facultade de
  Matem\'aticas, Universidade de Santiago de Compostela, Santiago De
  Compostela 15782, Spain}

\author[Bucataru]{Ioan Bucataru}
\address{Ioan Bucataru, Faculty of Mathematics, Universitatea
  ''Alexandru Ioan Cuza'' din Ia\c si, Ia\c si, 700506, Romania}
\urladdr{http://www.math.uaic.ro/\textasciitilde{}bucataru/}

\author[Salgado]{Modesto Salgado}
\address{Modesto Salgado, Departamento de Xeometr\'{\i}a e Topolox\'{\i}a, Facultade de
  Matem\'aticas, Universidade de Santiago de Compostela, Santiago De
  Compostela 15782, Spain}
\urladdr{http://xtsunxet.usc.es/modesto/}

\date{\today}

\begin{abstract}
In this paper we study symmetries, Newtonoid vector
  fields, conservation laws, Noether's Theorem and its converse, in the
  framework of the $k$-symplectic formalism, using the
  Fr\"olicher-Nijenhuis formalism on the space of $k^1$-velocities of
  the configuration manifold.

For the case $k=1$, it is well known that Cartan symmetries induce and
are induced by constants of motions, and these results are known as
Noether's Theorem and its converse. For the case $k>1$, we
  provide a new proof for
  Noether's Theorem, which shows that, in the
  $k$-symplectic formalism, each Cartan symmetry induces a conservation law.
We prove that, under some assumptions, the converse of Noether's Theorem
is also true and we provide examples when this is not the case. We also study the relations between
dynamical symmetries, Newtonoid vector fields, Cartan
symmetries and conservation laws, showing when one of them will imply the
others. We use several examples of partial differential equations to
illustrate  when these concepts are related and when they are not.
\end{abstract}

\subjclass[2000]{70S05, 70S10, 53D05}

\keywords{Symmetries, Conservation laws, Noether's theorem,
Lagrangian field theories, $k$-symplectic manifolds}

\maketitle

\section{Introduction}

The $k$-symplectic formalism \cite{ fam} is the generalization to field
theories of the standard symplectic formalism in
Mechanics \cite{am, arn}, which is the geometric framework for
describing autonomous dynamical systems. A natural extension of this
formalism is the so-called $k$-cosymplectic formalism, \cite{mod1,mod2}, which is a generalization to field theories of the
cosymplectic formalism describing non-autonomous mechanical systems. One of the advantages of using these formalisms is
 that  only the tangent and cotangent bundles of the configuration
 manifold are required to develop them. Others papers related with the
 $k$-symplectic and $k$-cosymplectic formalism are
 \cite{gun,LMMS1,LMMS2, mrsv,Mart2,relations, rsv07}.  

The polysymplectic formalism developed by
Giachetta, Mangiarotti and Sardanashvily in
\cite{Sarda2}, which is based on a vector-valued form defined
on some associated fiber  bundle, is a different description of
classical field theories of first order than the
$k$-symplectic  formalism. See also \cite{Kana}, for other
considerations regarding this aspect. The soldering form on linear frames bundle is a polysymplectic form, and its study and
applications to field theory, constitute the $n$-symplectic geometry
developed by Norris in \cite{No2,No3,No4,No5}. 

Alternatively, one can derive the field equations by use of the
so-called multisymplectic formalism, which was developed by
Tulczyjew's school, \cite{Kijo,KijoSz,KijTul,Snia,Tulczy1}, and
independently by Garc\'{\i}a and P\'{e}rez-Rend\'{o}n \cite{GP1,GP2}
and Goldschmidt and Sternberg \cite{GS}. This approach was revised
by   Gotay et al. \cite{Go1,Go2,Go3,Gymmsy} and more recently by Cantrijn et al.
\cite{Cant1,Cant2}. The relationship of the $k$-symplectic formalism with the multisymplectic formalism  is studied in \cite{relations}.

The aim of this paper is to study Noether's Theorem
 for first-order classical field theories, using the Lagrangian $k$-symplectic formalism.
 This study was initialized in \cite{rsv07} where large part of the
 discussion is a generalization of  the results obtained for
 non-autonomous mechanical systems. See, in  particular \cite{LM-96}
 and references quoted therein. We introduce the set of Newtonoid vector fields and prove that
 any Cartan symmetry is a Newtonoid vector field. Furthermore,
we show that under some assumptions, Newtonoid vector
  fields are Cartan symmetries and they induce conservation
  laws. This result extends the work developed by
  Marmo and Mukunda in \cite{marmo86}. The study of symmetries in field theory,
using various geometric frameworks, has been done in \cite{BSF88,
  EMR-99,Gymmsy, LMS-2004, mrsv}.

The structure of the paper is as follows. In Section \ref{review}  we review the $k$-symplectic
  Lagrangian formalism, and hence the field
theoretic state space of velocities is introduced in Section \ref{geo-elem} as the Whitney
sum $T^1_kQ$ of $k$-copies of the tangent bundle $TQ$ of a manifold $Q$. This manifold
has a canonical $k$-tangent structure defined by  $k$ tensor fields
$(J^1, \ldots, J^k)$ of type $(1,1)$, see
\cite{mt1,mt2}. In the case $k=1$, $J^1$ is the
canonical tangent structure of the
tangent bundle $TQ$. The canonical $k$-tangent structure of $T^1_kQ$ is used
to construct the Poincar\'{e}-Cartan forms.

A particular type of second order partial differential equations,
which we call {\sc sopde}, are introduced in Section \ref{kvf}. They are a
generalization of {\sc sode's} (semisprays) found in Geometric Mechanics. The Lagrangian
formalism  is developed  in Section \ref{subel}.

In  Section \ref{nth} we discuss symmetries and conservation laws
 for Lagrangian functions on $T^1_kQ$. We prove Noether's Theorem \ref{thm:noether}, which
 shows that each Cartan symmetry induces a conservation law. We provide in Proposition \ref{convnt} some
conditions under which the converse of Noether's Theorem is true. Noether's Theorem \ref{thm:noether}
was proved previously in \cite{rsv07} using local coordinates. Here we present a direct
global proof using the Fr\"olicher-Nijenhuis formalism. For a modern
description of the Fr\"olicher-Nijenhuis formalism see \cite[\S 8]{KMS93}.
In Section \ref{subnewt} we introduce the set of Newtonoid vector
fields in the framework of $k$-symplectic formalism, extending the work of Marmo and Mukunda \cite{marmo86}
  for the case $k=1$. In Proposition \ref{prop:csn} we prove that Cartan
  symmetries are always Newtonoid vector fields. In Theorem
  \ref{thm:cscl} we show that, under some assumptions, Newtonoid
  vector fields are Cartan symmetries and hence they provide
  conservation laws.

\section{Review of Lagrangian $k$-symplectic formalism}\label{review}
In this section we briefly recall the Lagrangian $k$-symplectic
formalism. We refer the reader to \cite{fam,rsv07} for more details about this
formalism. We present first the geometric framework
  for this formalism, which is given by the tangent bundle of
  $k^1$-velocities $T^1_kQ$ of the configuration manifold $Q$, together with the
  canonical structures. For a Lagrangian on $T^1_kQ$, the
  geometric informations we need for the Lagrangian $k$-symplectic
  formalism are encoded in the Poincar\'e-Cartan forms. We discuss
  further systems of second order partial differential equations ({\sc
    sopde}) as well as their relations with Euler-Lagrange equations.

\subsection{Geometric framework}\label{geo-elem}

\subsubsection*{The tangent bundle of $k^1$-velocities of a manifold.
Canonical structures.}

  In this work we consider $Q$ a real, $n$-dimensional
  and $C^{\infty}$-smooth manifold. Throughout the paper, we assume that all
  objects are $C^{\infty}$-smooth where defined. Consider $(TQ, \tau,
  Q)$ the tangent bundle of the manifold $Q$. We denote by $C^{\infty}(Q)$ the ring of smooth
  functions on $Q$, and by $\mathfrak{X}(Q)$ the
  $C^{\infty}(Q)$-module of vector fields on $Q$.

Let us denote by $T^1_kQ$ the Whitney sum
$TQ\oplus\stackrel{k}{\dots}\oplus TQ$ of $k$ copies of $TQ$, with
projection $\tau_Q\colon T^1_kQ\to Q$. $T^1_kQ$ can be identified with
the manifold $J^1_0(\mathbb{R}^k,Q)$ of the $k^1$-velocities of $Q$;
that is,  $1$-jets of maps $\sigma\colon\mathbb{R}^k\to Q$, with
source at $0\in \mathbb{R}^k$. For this reason the manifold $T^1_kQ$
is called \emph{the tangent bundle of
$k^1$-velocities of $Q$}, see \cite{morimoto}. If $(q^i)$ are local
coordinates on $U \subseteq Q$, then
the induced local coordinates  in
$T^1_kU=\tau_Q^{-1}(U)$ are denoted by $(q^i ,
v_\alpha^i)$, $ 1\leq i\leq n$,
$1\leq \alpha \leq k$. Throughout this work we implicitly assume summation over repeated
covariant and contravariant latin indices $i, j, l, ... \in
\{1,...,n\}$, as well as summation over repeated greek indices $\alpha, \beta,
... \in \{1,...,k\}$.

The \emph{canonical $k$-tangent structure} on $T^1_kQ$, see \cite{fam},
is the family
$J=(J^1,\ldots, J^k)$ of $k$ tensor fields  of type $(1,1)$, which are
locally given by
\begin{equation}
\label{localJA} J^\alpha=\frac{\partial}{\partial v^i_\alpha} \otimes
d q^i, \quad \alpha \in \{1,...,k\}.
\end{equation}
 In the case
$k=1$, $J^1$ is the well-known canonical tangent structure of the
tangent bundle.

The \emph{Liouville vector field}, $\mathbb{C}\in\mathfrak{X}(T^1_kQ)$,
is the infinitesimal generator of the following flow
 $$
\psi\colon \mathbb{R}\times T^1_kQ\longrightarrow T^1_kQ, \quad
  \psi(s, (q, v_{1}, \ldots ,v_{k}))= (q, e^s v_{1},\ldots , e^s  v_{k}).
$$
In local coordinates, the Liouville vector field has the form
\begin{equation}\label{localC}
 \mathbb{C} = v^i_\alpha \frac{\partial}{\partial v^i_{\alpha}} .
\end{equation}

The \emph{vertical distribution} is the $kn$-dimensional distribution
on $T^1_kQ$ given by $V: u\in T^1_kQ \to V(u)=\operatorname{Ker}d_u\tau_Q
= \operatorname{Ker} J_u \subset
T_uT^1_kQ$. The vertical distribution $V$ splits into $k$
subdistributions $V^{\alpha}(u)=\operatorname{Im} J^{\alpha}(u)$,
$\alpha\in \{1,...,k\}$. Each of these vertical subdistributions are
$n$-dimensional and integrable since
$V^{\alpha}(u)=\textrm{span}\left\{\partial/\partial
  v^i_{\alpha},  1\leq i\leq n \right\}$. 
% We denote by
%  $\mathfrak{X}^{V}(T^1_kQ)$ the Lie algebra of vertical vector fields
%  on $T^1_kQ$.

\subsubsection*{ Poincar\'{e}-Cartan  forms on $T^1_kQ$.} \label{pcforms}

The Lagrangian $k$-symplectic formalism for a Lagrangian function $L$ on $T^1_kQ$ can be developed from the corresponding Poincar\'e-Cartan forms.

\begin{defn}   A  \emph{Lagrangian} is a smooth function $L$ on
  $T^1_kQ$. A Lagrangian $L \in C^{\infty}(T^1_kQ)$ is called
  \emph{regular} if the Hessian matrix of $L$ with respect to the fibre
  coordinates,
\begin{eqnarray}
g^{\alpha\beta}_{ij}(q, v)=\frac{\partial^2 L}{\partial
  v^i_{\alpha} \partial v^j_{\beta}}(q,v) \label{gabij} \end{eqnarray} has
maximal rank $kn$ on $T^1_kQ$.

For a Lagrangian $L$, the \emph{energy function} is
$E_L=\mathbb{C}(L)-L\in C^\infty(T^1_kQ)$, with local expression
\begin{eqnarray} E_L=v^i_\alpha\frac{\partial L}{\partial v^i_\alpha}-L .
\label{energyL} \end{eqnarray}
\end{defn}

For each Lagrangian $L\in C^{\infty}(T^1_kQ)$ we
consider the family of \emph{Poincar\'{e}-Cartan} $1$-forms
$\theta_L^\alpha=d_{J^{\alpha}}L= dL \circ J^{\alpha}, $ as well as
the family of Poincar\'{e}-Cartan $2$-forms on $T^1_kQ$,
$\omega_L^\alpha=-d\theta_L^\alpha$.

In induced local coordinates on $T^1_kQ$, the Poincar\'{e}-Cartan
forms are given by \begin{eqnarray}
\label{thetala} \theta_L^\alpha &=& \frac{\partial L}{\partial
v^i_\alpha} dq^i,
 \\
\nonumber\omega_L^\alpha &=& \frac{1}{2} \left(\frac{\partial ^2
L}{\partial q^j\partial v^i_\alpha} - \frac{\partial ^2 L}{\partial
q^i\partial v^j_\alpha} \right) dq^i\wedge dq^j + \frac{\partial ^2
L}{\partial v^j_\beta\partial v^i_\alpha} dq^i\wedge d v^j_{\beta}.
\end{eqnarray}

We recall now the definition of a $k$-symplectic structure, see \cite{aw1, aw3}.
\begin{defn} \label{definition:ksympecticl}
A \emph{$k$-symplectic} structure on a  $k+nk$-dimensional manifold
$M$  is given by a family of $k$, closed $2$-forms $(\omega^1,...,
\omega^k)$ and an integrable $kn$-dimensional distribution $V$    on $M$ such
that $$ i) \ \cap_{\alpha=1}^k
  \operatorname{Ker}(\omega^{\alpha})=\{0\}, \quad ii) \
  \omega^{\alpha}_{\vert_{ V\times V}}=0, \quad  \alpha\in
  \{1,...,k\}\, . $$
\end{defn}

Using formulae   \eqref{gabij} and \eqref{thetala} one obtains
that a Lagrangian $L$ is regular if and only if the
Poincar\'e-Cartan $2$-forms and the vertical
  distribution, $(\omega_L^1, \ldots, \omega_L^k, V= \ker \tau_Q)$,
  define a $k$-symplectic structure on $T^1_kQ$, see \cite{fam}.

\subsubsection*{Complete lifts of vector fields.} The
  lifting process of some geometric structures from a base manifold to
  the total space of some fibre bundle has proven its usefulness for
  studying the corresponding geometric structures \cite{morimoto, yano73}. For example,
  the complete lift of a system of second order ordinary differential
  equations contains informations about its symmetries and first order
  variations, \cite{bd10}.
\begin{defn} Let $\phi\colon Q\to Q$ be a differentiable map, then the
\emph{first order prolongation} of $\phi$ to $T^1_kQ$ is the
map $T^1_k\phi:T^1_kQ \to  T^1_kQ$, defined by
$T^1_k\phi(j^1_0\sigma)=j^1_0(\phi \circ \sigma)$. This means that
for $(q, {v_1}, \ldots, {v_k})\in T_qQ$, $q\in Q$, we have
$$
  T^1_k\phi(q, {v_1},\ldots, {v_k})=
(d_q\phi (v_1), \ldots,d_q\phi (v_k)).
$$
\end{defn}

If $Z$ is a vector field on $Q$, with local $1$-parametric group of
transformations $h_s\colon Q \to Q$, then the local $1$-parametric
group of transformations $T^1_kh_s\colon T^1_kQ \to T^1_kQ$
generates a vector field $Z^C$ on $T^1_kQ$, which is called the
\emph{complete lift} of $Z$ to $ T^1_kQ$. If locally $Z=Z^i{\partial}/{\partial q^i}$, then the complete lift is given by
\begin{eqnarray*} \label{clift}
Z^C=Z^i\frac{\partial}{\partial q^i} + v^i_\alpha\frac{\partial Z^j}
{\partial q^i} \frac{\partial}{\partial v_\alpha^j}.
\end{eqnarray*}
 If we consider also the vertical lifts $Z^{V_{\alpha}}=J^{\alpha}Z^C$,
then the following properties  are well known, see \cite{morimoto},
$$
[X^C,Y^C]= [X,Y]^C, \quad [X^C,Y^{V_\alpha}]= [X,Y]^{V_\alpha},
\quad  [X^{V_\alpha},X^{V_\beta}]=0.
$$ These formulae extend the well known properties of Lie brackets
for vertical and complete lifts of vector fields to $TQ$,
\cite{yano73}.

\subsection{Systems of first and second-order partial differential
  equations}\label{kvf}

A vector field on a manifold $M$ defines a system of first-order ordinary differential equations. Accordingly, a $k$-vector field on
$M$, for some $k>1$, defines  a system of first-order partial differential equations. Furthermore, some special $k$-vector field on the manifold
$M=T^1_kQ$ defines a system of second-order partial differential equations.

\subsubsection*{First-order partial differential equations on a manifold}

In this subsection, we briefly show how $k$-vector
fields determine systems of first-order partial differential
equations. 

\begin{defn}
 A  \emph{$k$-vector field} on an arbitrary manifold $M$ is a section
  $X: M \longrightarrow T^1_kM$
of the canonical projection $\tau_M:T^1_kM \rightarrow M$.
\end{defn}

Since $T^{1}_{k}M$ is  the Whitney sum $TM\oplus \stackrel{k}{\dots}
\oplus TM$ of $k$ copies of $TM$, we deduce that a $k$-vector field
$X$ defines a family of $k$ vector fields $X_{1}, \dots,
X_{k}\in\mathfrak{X}(M)$ by projecting $X$ onto every factor; that
is, $X_\alpha=\tau_{\alpha}\circ X$, where $\tau_{\alpha}\colon
T^1_kM \rightarrow TM$ is the canonical projection on the
$\alpha^{th}$-copy $TM$ of $T^1_kM$.

\begin{defn}\label{integsect}
An \emph{integral section}  of the $k$-vector field $X=(X_{1},
\dots, X_{k})$, passing through a point $x\in M$, is a map
$\psi\colon U\subset {\mathbb R}^{k} \rightarrow M$, defined on some
open neighborhood  $U$ of $0\in {\mathbb R}^{k}$, such that
$$
\psi(0)=x, \ d_t\psi\left(\frac{\partial}{\partial
t^\alpha}\Big\vert_t\right)=X_{\alpha}(\psi (t))\in T_{\psi(t)}M,
\quad \rm{for \ every\ } t\in U, \ 1\leq \alpha\leq k.
$$
Equivalently, $  \psi  $ satisfies $X\circ \psi =
\psi^{(1)}$, where $\psi^{(1)}: U\subset {\mathbb
  R}^{k} \longrightarrow  T^1_kM$    is the
first-order prolongation of $  \psi$ to $T^1_kM$ defined by
\begin{equation}
\begin{array}{rcl}
\psi^{(1)}: U\subset {\mathbb
  R}^{k} & \rightarrow  & T^1_kM \\ \noalign{\medskip}
       t & \rightarrow  & \psi^{(1)}(t)=
       \left(d_t\psi\left(\frac{\partial}{\partial t^1}\Big\vert_t\right),\ldots,
d_t\psi\left(\frac{\partial}{\partial t^k}\Big\vert_t\right)\right).
\end{array}
\end{equation}

In local coordinates, if $\psi(t)=(\psi^i(t))$, then
we have
\begin{equation}\label{localfi11}
\psi^{(1)}(t)= \left( \psi^i (t),
  \frac{\partial\psi^i}{\partial t^\alpha} (t)\right).
\end{equation}
   A $k$-vector field $ X=(X_1,\ldots , X_k)$ on $M$ is
\emph{integrable} if there is an integral section passing through
every point of $M$.
\end{defn}

Consider $X=(X_{1}, \dots, X_{k})$ a $k$-vector field, where
$X_\alpha=X^i_{\alpha} {\partial }/{\partial x^i}$ in a  coordinate
system $(U,x^i)$ on  $M$. The $k$-vector field $X$ induces a system of first-order
partial differential equations on $M$, which is given by
$$
X_\alpha^i(x^j(t))= \frac{\partial x^i}{\partial t^\alpha}\Big\vert_t,
\quad   \alpha \in \{1,..., k\}, \quad  i\in \{1,..., \dim M\}.
$$
From Definition \ref{integsect} we deduce that $\psi$ is an integral section of
$X=(X_{1}, \dots, X_{k})$
if $\psi$ is a solution to the above system of first-order
partial differential equations, which means that it satisfies
$$
X_\alpha^i(\psi(t))= \frac{\partial \psi^i}{\partial t^\alpha}\Big\vert_t,
\quad   \alpha \in \{1,..., k\},  \quad i\in \{1,..., \dim M\}.
$$

\subsubsection*{Systems of second-order partial differential
equations}\label{sopde}

In this part we characterize those integrable
$k$-vector fields on $M=T^1_kQ$ that have as integral sections first
order prolongations $\phi^{(1)}$ of maps $\phi: U\subset {\mathbb
R}^{k}\to Q$. Such $k$-vector fields define integrable systems of
second-order partial differential equations on the base manifold
$Q$.

As we recalled, a $k$-vector field in $T^1_kQ$ is a section
$\mathbf{\xi}\colon T^1_kQ\longrightarrow T^1_k(T^1_kQ)$ of the
canonical projection $\tau_{T^1_kQ}\colon T^1_k(T^1_kQ)\to
T^1_kQ$. We note that there are systems of partial
  differential equations that are not induced by $k$-vector
  fields. However, in this work we are interested only in those systems
  of PDE that are induced by such $k$-vector fields. In view of these
  considerations, we consider the following definition.
\begin{defn}
\label{sode0} A system of \emph{second-order partial differential equations}
({\sc sopde}) on $Q$ is a $k$-vector field
$\xi=(\xi_1,\ldots,\xi_k)$ on $T^1_kQ$, which is
a section of the projection $T^1_k\tau_Q\colon
T^1_k(T^1_kQ)\rightarrow T^1_kQ$, namely
$ T^1_k\tau_Q\circ\mathbf{\xi}={\rm Id}_{T^1_kQ}, $
and this is equivalent to
\begin{equation}\label{propsopde}
d\tau_Q\circ \xi_\alpha=\tau_{\alpha}:T^1_kQ \to TQ, \quad  \alpha \in \{1,..., k\}.
\end{equation}
\end{defn}
Equivalently, above equations can be written as follows
$$
d_{(q,v)}\tau_Q(\xi_\alpha(q,v))= (q,v_{\alpha}),\quad {\rm for \ }
(q, v_1,..., v_k)\in T^1_kQ, \quad \alpha \in \{1,\ldots, k\}.
$$
In the case $k=1$, Definition \ref{sode0} reduces to the definition of a system of second-order
ordinary differential equations ({\sc sode}).

Locally, a {\sc sopde} $\mathbf{\xi}=(\xi_1,\ldots,\xi_k) $ is given
by
\begin{equation}
\label{localsode1} \xi_\alpha= v^i_\alpha\frac{\partial}
{\partial q^i}+ \xi^i_{\alpha \beta} \frac{\partial} {\partial
v^i_\beta},\quad  \alpha \in \{1,..., k\},
\end{equation}
where $\xi^i_{\alpha \beta}$ are smooth functions defined on domains
of induced charts on $T^1_kQ$.

All these considerations allow us to reformulate the definition for a
 {\sc sopde}, using the $k$-tangent structure   and the Liouville vector field
 $\mathbb{C}$ (see formulae \eqref{localJA} and \eqref{localC}), as follows.

 \begin{prop}\label{xijso}
 A $k$-vector field
$\xi=(\xi_1,\ldots,\xi_k)$ on $T^1_kQ$ is a  {\sc sopde} if and only
if $J^{\alpha}(\xi_{\alpha})=\mathbb{C}$.
 \end{prop}

If $\psi: U\subset {\mathbb R}^{k} \to T^1_kQ$,  locally given by
$\psi(t)=(\psi^i(t),\psi_\alpha^i(t))$, is an integral section of a {\sc sopde}
$\mathbf{\xi}=(\xi_1,\ldots,\xi_k)$ then from Definition
\ref{integsect} and formula \eqref{localsode1} it follows
\begin{equation}
\label{solsopde}
 \frac{\partial\psi^i} {\partial t^\alpha}\Big\vert_t=\psi^i_\alpha(t),\quad
\frac{\partial\psi^i_\alpha} {\partial
t^\beta}\Big\vert_t=\xi^i_{\alpha \beta}(\psi(t)).
\end{equation}
Using formulae \eqref{localfi11} and \eqref{solsopde} we obtain the
following characterization for the integral maps of a {\sc sopde}.
\begin{prop}
Let $\xi=(\xi_1,\ldots,\xi_k)$ be an integrable
{\sc sopde}. If $\psi$ is an integral section of $\xi$,
then $\psi=\phi^{(1)}$, where $\phi^{(1)}$ is the first-order
prolongation of the map
$$\phi=\tau_Q\circ\psi: U\subset {\mathbb R}^{k}\stackrel{\psi}{\to}T^1_kQ
\stackrel{\tau_Q}{\to}Q,$$
and $\phi$ is a solution to the system of second-order partial
differential equations
\begin{equation}
\label{nn1}
 \frac{\partial^2 \phi^i}{\partial t^\alpha\partial t^\beta}(t)=
\xi^i_{\alpha \beta}\left(\phi^j(t), \frac{\partial\phi^j}{\partial
t^\gamma}(t)\right).
\end{equation}
Conversely, if $\phi: U\subset {\mathbb R}^{k} \to Q$ is any map satisfying
the system \eqref{nn1}, then $\phi^{(1)}$ is an integral section of
$\mathbf{\xi}=(\xi_1,\ldots,\xi_k)$.
\end{prop}
\begin{defn}
If $\phi^{(1)}$ is an integral section of a {\sc sopde}
$\xi=(\xi_1,\ldots,\xi_k)$, then map $\phi$ will be called a \emph{solution} to $\xi$.
\end{defn}

From equations \eqref{nn1} we deduce that if $\xi$ is an
integrable {\sc sopde} then necessarily we have the symmetry
$\xi^i_{\alpha \beta}=\xi^i_{ \beta\alpha}$ for all $\alpha ,\beta=1,\ldots, k$.
Therefore, for a {\sc sopde} $\xi$, locally given by  formula \eqref{localsode1},
 we require the following integrability conditions \cite{kosambi48}:
\begin{eqnarray} \xi^i_{\alpha\beta}=\xi^i_{\beta\alpha}, \quad
\xi_{\alpha}(\xi^i_{\beta\gamma})=\xi_{\beta}(\xi^i_{\alpha\gamma}), \forall
\alpha, \beta, \gamma \in\{1,...,k\}. \label{symcond} \end{eqnarray}
Integrability conditions \eqref{symcond} are equivalent to the fact
that $[\xi_{\alpha}, \xi_{\beta}]=0$, $\forall
\alpha, \beta \in\{1,...,k\}$. The integrability conditions \eqref{symcond} have been also proved in \cite{Maranon}.  Due to the first symmetry condition
\eqref{symcond} we have that the system \eqref{nn1}
is a system of $nk(k+1)/2$ second-order partial differential equations.

\subsection{Euler-Lagrange equations}\label{subel}

An important class of {\sc sopde} on a manifold $Q$ contains those whose
  solutions are among the solutions of the Euler-Lagrange equations
  for some Lagrangian function on $T^1_kQ$. In Proposition
  \ref{prop:elsopde} we characterize this class, while in Proposition
  \ref{known2} we discuss the relation between the solutions of a {\sc
  sopde} in this class and the solutions of the corresponding
Euler-Lagrange equations.

The variational problem for a Lagrangian $L$ on $T^1_kQ$ leads to the
following system of Euler-Lagrange equations
\begin{eqnarray}
\frac{\partial}{\partial t^{\alpha}}\left(\frac{\partial
    L}{\partial v^i_{\alpha}}\right) - \frac{\partial L}{\partial
  q^i}=0.
\label{eqn:eulerlagrange} \end{eqnarray}
Euler-Lagrange equations \eqref{eqn:eulerlagrange} can be written as
\begin{eqnarray}
g^{\alpha\beta}_{ij} \frac{\partial^2 q^j}{\partial t^{\alpha}\partial
  t^{\beta}} + \frac{\partial^2L}{\partial q^j \partial v^i_{\alpha}} v^j_{\alpha}
- \frac{\partial L}{\partial q^i}=0,  \label{eqn:eulerlagrange2} \end{eqnarray}
which represents a system of $n$ second-order partial differential
equations on $Q$.

Denote by $\mathfrak{X}^k_L(T^1_kQ)$ the set of $k$-vector fields
$\xi=(\xi_1,\dots,\xi_k)$ on $T^1_kQ$, which are solutions
to the equation
 \begin{equation}
\label{genericEL} i_{\xi_{\alpha}} \omega_L^\alpha=d E_L.
 \end{equation}

If each  $\xi_{\alpha}$ is locally given by
\begin{equation}
\label{localsode2} \xi_\alpha= \xi^i_\alpha\frac{\partial}
{\partial q^i}+ \xi^i_{\alpha \beta} \frac{\partial} {\partial
v^i_\beta},\quad  \alpha \in \{1,..., k\},
\end{equation}
then $(\xi_1,\dots,\xi_k)$ is a solution to \eqref{genericEL}
if and only if the functions  $ \xi^i_{\alpha}$ and $\xi^i_{\alpha\beta}$
satisfy the following system of equations
\begin{eqnarray}\label{ellnr}
  \left( \frac{\partial^2 L}{\partial q^i \partial v^j_{\alpha}} -
 \frac{\partial^2 L}{\partial q^j \partial v^i_{\alpha}}
\right) \xi_{\alpha}^j - \frac{\partial^2 L}{\partial v_\alpha^i
\partial v^j_{\beta}}  \xi_{\alpha\beta}^j &=&
v_\alpha^j \frac{\partial^2 L}{\partial q^i
\partial v^j_{\alpha}} -\frac{\partial  L}{\partial q^i }
\\ \nonumber \frac{\partial^2 L}{\partial v^j_{\beta}\partial v^i_{\alpha}}  \xi_{\alpha}^i
&=& \frac{\partial^2 L}{\partial v^j_{\beta}\partial v^i_{\alpha}}  v_\alpha^i.
\end{eqnarray}
If $L$ is a regular Lagrangian, the above equations are equivalent to the following  equations
\begin{equation}\label{locel4}
\frac{\partial^2 L}{\partial v_\alpha^i \partial v^j_{\beta}}
\xi_{\alpha\beta}^j  + \frac{\partial^2 L}{\partial q^j \partial v^i_{\alpha}} v^j_{\alpha}
- \frac{\partial  L}{\partial q^i}=0,\quad \xi^i_{\alpha}= v^i_{\alpha}.
\end{equation}

Using equations \eqref{ellnr} we deduce the following lemma.
\begin{lem}\label{know1} Consider $L\in
    C^{\infty}(T^1_kQ)$ a Lagrangian.
\begin{enumerate}
\item[1)]  If $L$ is regular,  then any $k$-vector field $\xi\in
\mathfrak{X}^k_L(T^1_kQ)$ is a {\sc sopde}, it is locally given by formula
\eqref{localsode1} and satisfies equations \eqref{locel4}.

\item[2)] If $\xi \in \mathfrak{X}^k_L(T^1_kQ)$ and $\xi$ is a {\sc sopde}, then
it is locally given by formula \eqref{localsode1} and satisfies equations \eqref{locel4}.
\end{enumerate}
\end{lem}

Next proposition characterizes the set of {\sc sopde}s that are in $\mathfrak{X}^k_L(T^1_kQ)$.

\begin{prop} \label{prop:elsopde}
Let $L\in C^\infty(T^1_kQ) $ be  a Lagrangian and let
$\xi=(\xi_1,\dots,\xi_k)$  be a {\sc sopde} on $T^1_kQ$. Then
 $\xi \in
\mathfrak{X}^k_L(T^1_kQ)$ if and only if it satisfies  the
following   condition:
\begin{eqnarray}
\label{eq}  \mathcal{L}_{\xi_\alpha} \theta^\alpha_L = dL, \quad \mbox{or locally} \quad
\xi_\alpha\left(\frac{\partial L}{\partial v^i_\alpha}\right)=
\frac{\partial L}{\partial q^i}.
\end{eqnarray}
\end{prop}

\begin{proof}
We will start by proving that the first equation in \eqref{eq} is equivalent to
equation \eqref{genericEL}. Since $\xi$ is a {\sc sopde} we have that
$J^{\alpha}\xi_{\alpha}=\mathbb{C}$ and hence it follows that
$i_{\xi_{\alpha}}\theta^{\alpha}_L= \theta^\alpha_L(\xi_\alpha)=(dL\circ
J^{\alpha})(\xi_{\alpha})=\mathbb{C}L$.  Using the fact that
$\omega^\alpha_L=-d\theta^\alpha_L$ we obtain
$$\mathcal{L}_{\xi_\alpha} \theta^\alpha_L=di_{
  \xi_\alpha}\theta^\alpha_L+ i_{\xi_\alpha}d\theta^\alpha_L =
d(\mathbb{C}L) - i_{ \xi_\alpha}\omega^\alpha_L = dL +\left(dE_L -i_{\xi_\alpha}\omega^{\alpha}_L\right) .
$$
It follows that the first equation in \eqref{eq} is equivalent to
equation \eqref{genericEL}.

Since $\xi$ is a {\sc sopde} it follows that
$\xi_\alpha(q^i)=v^i_\alpha$.  Therefore, we have
\begin{eqnarray*}
\mathcal{L}_{\xi_\alpha}\theta^\alpha_L - dL &=& \mathcal{L}_{\xi_\alpha}
\left( \frac{\partial L}{\partial v^i_{\alpha}} dq^i\right) - dL
= \xi_{\alpha}\left( \frac{\partial L}{\partial v^i_{\alpha}} \right)
dq^i + \left( \frac{\partial L}{\partial v^i_{\alpha}} \right)
dv^i_{\alpha} - dL \\ &=& \left[\xi_{\alpha}\left( \frac{\partial
      L}{\partial v^i_{\alpha}} \right) - \frac{\partial L}{\partial
    q^i}\right] dq^i,
\end{eqnarray*}
and hence the two equations in \eqref{eq} are equivalent.  \end{proof}

We will discuss now the relation between solutions of the
Euler-Lagrange equations \eqref{eqn:eulerlagrange} or
\eqref{eqn:eulerlagrange2} and
integral sections of $k$-vector fields
in $\mathfrak{X}^k_L(T^1_kQ)$.
\begin{prop} \label{known2}  Consider a Lagrangian $L$ on $T^1_kQ$ and a $k$-vector
  field $\xi \in \mathfrak{X}^k_L(T^1_kQ)$.
\begin{enumerate}
\item[1)] If $\xi$ is a {\sc sopde}, then a map $\phi: U\subset {\mathbb
    R}^k \to  Q$ is a solution to the
  Euler-Lagrange equations \eqref{eqn:eulerlagrange} if and only if
\begin{eqnarray}
g^{\alpha\beta}_{ij}\circ \phi^{(1)} \left( \xi^j_{\alpha\beta}\circ
  \phi^{(1)} - \frac{\partial^2 \phi^j}{\partial t^{\alpha}\partial
      t^{\beta}}\right)=0. \label{xiel}
\end{eqnarray}
\item[2)] If the Lagrangian $L$ is regular, then $\xi$ is a {\sc sopde}, and
  if $\phi: U\subset {\mathbb R}^k \to  Q$ is a solution to $\xi$,
  then $\phi$ is a solution to the Euler-Lagrange equations \eqref{eqn:eulerlagrange}.
\item[3)] If $\xi$ is integrable,
and $\phi^{(1)}: U\subset {\mathbb R}^k \to   T^1_kQ$ is an integral section,
then $\phi: U\subset {\mathbb R}^k \to Q$ is a solution to the Euler-Lagrange
equations \eqref{eqn:eulerlagrange}.
\end{enumerate}
\end{prop}
\begin{proof}
1) Consider a map $\phi: U\subset {\mathbb
    R}^k \to  Q$. If $\phi$ is a solution to the
   Euler-Lagrange equations
  \eqref{eqn:eulerlagrange2},
  then we have
 \begin{eqnarray}
g^{\alpha\beta}_{ij}\circ \phi^{(1)} \frac{\partial^2 \phi^j}{\partial t^{\alpha}\partial
  t^{\beta}} + \frac{\partial^2L}{\partial q^j \partial v^i_\alpha}
\circ \phi^{(1)}  \frac{\partial \phi^j}{\partial t^\alpha }
- \frac{\partial L}{\partial q^i}\circ \phi^{(1)}=0.
\label{gxiphi1}\end{eqnarray}
If the $k$-vector field $\xi$ is a  {\sc sopde}, then $\xi \in
\mathfrak{X}^k_L(T^1_kQ)$  if and only if it satisfies the equations
\begin{eqnarray}
\frac{\partial^2 L}{\partial q^j \partial v^i_{\alpha}}\ v^j_{\alpha} +
\frac{\partial^2 L}{\partial v_\alpha^i \partial v^j_{\beta}}
\xi_{\alpha\beta}^j =
\frac{\partial  L}{\partial q^i}.
\label{gxiphi0}\end{eqnarray}
If we  restrict equation \eqref{gxiphi0} to the image of $\phi^{(1)}$ we
  obtain
\begin{eqnarray} g^{\alpha\beta}_{ij}\circ \phi^{(1)} \xi^j_{\alpha\beta}\circ
  \phi^{(1)} + \frac{\partial^2 L}{\partial q^j \partial v^i_{\alpha}}
  \circ \phi^{(1)} \frac{\partial \phi^j}{\partial t^{\alpha}} - \frac{\partial L}{\partial q^i}
  \circ \phi^{(1)}
  =0.\label{gxiphi}\end{eqnarray}
 Using equations \eqref{gxiphi} it follows that $\phi$ satisfies \eqref{xiel}
 if and only if it satisfies \eqref{gxiphi1} that are equivalent to
 Euler-Lagrange equations \eqref{eqn:eulerlagrange2}.

2) If $\phi: U\subset {\mathbb R}^k \to  Q$ is a solution to $\xi$
then it satisfies equations \eqref{nn1}. Therefore, equations
\eqref{xiel} are automatically satisfied and hence $\phi$ is a
solution of the Euler-Lagrange equations \eqref{eqn:eulerlagrange2}.

3) Since $\xi \in \mathfrak{X}^k_L(T^1_kQ)$ it follows that $\xi$
satisfies first equation \eqref{ellnr}. If we restrict this equation
to $\phi^{(1)}: U\subset {\mathbb R}^k \to   T^1_kQ$, which is an integral map
of $\xi$, we obtain that $\phi$ satisfies the Euler-Lagrange equations
\eqref{eqn:eulerlagrange2}. \end{proof}

\begin{rem}
The results of Lemma \ref{know1} and results $2)$ and $3)$ of Proposition \ref{known2} are the fundaments of Lagrangian $k$-symplectic formalism and equation \eqref{genericEL} can be seen as a geometric version of the Euler-Lagrange field equations.
\end{rem}

\begin{rem}
Formula \eqref{xiel} does not require any
relationship between the $k$-vector field $\xi \in
\mathfrak{X}^k_L(T^1_kQ)$ and the solution $\phi$ to the Euler-Lagrange equations
\eqref{eqn:eulerlagrange}.  In other words, we might
  have $\phi$ a solution to the Euler-Lagrange equations
  \eqref{eqn:eulerlagrange} which may not be a solution for any $\xi \in \mathfrak{X}^k_L(T^1_kQ)$
\end{rem}

\begin{ex} In this example we consider the theory of a
vibrating string. Coordinates $(t^1,t^2)$  are interpreted as the
time and the distance along the string, respectively. If
$
 \phi: (t^1,t^2)\in  \mathbb{R}^2  \longrightarrow  \phi(t^1,t^2)\in \mathbb{R}
$ denotes the displacement of each point of the
string as function of the time $t^1$ and the position $t^2$, the
motion  equation is
\begin{equation}\label{vs}
\sigma\frac{\displaystyle\partial^2 \phi}{\displaystyle\partial
(t^1)^2}-\tau \frac{\displaystyle\partial^2
\phi}{\displaystyle\partial (t^2)^2}=0,
\end {equation}
where $\sigma$ and $\tau$ are certain constants of the mechanical
system.

 Equation \eqref{vs} is the Euler-Lagrange equation
 for the regular Lagrangian
\begin{eqnarray} \label{lvs}
L: T^1_2 \mathbb{R} \to  \mathbb{R}, \
  L(q,v_1,v_2)=\frac{\displaystyle 1}{\displaystyle 2}(\sigma v^2_1-\tau v^2_2).\end{eqnarray}
From formulae \eqref{thetala} and \eqref{lvs} we deduce that
\begin{eqnarray} \label{locomel}
     \omega_L^1  =   \sigma dq\wedge dv_1, \quad  \omega_L^2 = -\tau dq \wedge dv_2, \quad
     dE_L=\sigma v_1dv_1-\tau v_2dv_2.
\end{eqnarray}
Therefore, a {\sc sopde} $(\xi_1, \xi_2)\in
\mathfrak{X}(T^1_2 \mathbb{R})$
$$
\xi_1=v_1\frac{\partial}{\partial q} +  \xi_{11} \frac{\partial}{\partial v_1} + \xi_{12}
\frac{\partial}{\partial v_2}, \quad
\xi_2=v_2\frac{\partial}{\partial q} + \xi_{12} \frac{\partial}{\partial
  v_1} +  \xi_{22} \frac{\partial}{\partial v_2},
$$
 is a solution to equation \eqref{genericEL}  if and only if it satisfies
 \begin{eqnarray}  \sigma \xi_{11} -\tau  \xi_{22}
=0.\label{geoeqej} \end{eqnarray}
The integrability conditions \eqref{symcond} are in this case
\begin{eqnarray}
\frac{\partial \xi_{11}}{\partial v_1} = \frac{\partial
  \xi_{12}}{\partial v_2},  \quad \sigma \frac{\partial
  \xi_{11}}{\partial v_2} = \tau \frac{\partial
  \xi_{12}}{\partial v_1}.  \label{diffsvs}
\end{eqnarray}
An example of an integrable {\sc sopde}, which is a solution to \eqref{geoeqej}
  is given by
\begin{eqnarray*}
\xi_1=v_1\frac{\partial}{\partial q} + \tau \left( \sigma (v_1)^2 + \tau
(v_2)^2 \right) \frac{\partial}{\partial v_1} + 2\sigma \tau v_1v_2
\frac{\partial}{\partial v_2}, \\
\xi_2=v_2\frac{\partial}{\partial q} + 2\sigma \tau v_1v_2
\frac{\partial}{\partial v_1} + \sigma \left( \sigma (v_1)^2 + \tau
(v_2)^2\right)  \frac{\partial}{\partial v_2}. \label{xivs}
\end{eqnarray*}
Thus any solution $\phi$ of the integrable {\sc sopde} $(\xi_1,
  \xi_2)$ in the formulae above is a solution of the vibrating string
  equation \eqref{vs}.
\end{ex}

\section{Noether's theorem}\label{nth}

In this section we discuss symmetries and conservation laws
 for Lagrangian functions on $T^1_kQ$. We introduce
the Newtonoid vector fields in this framework,
 extending the work of Marmo and Mukunda \cite{marmo86}
  for the case $k=1$. We provide a new proof for Noether's Theorem
  \ref{thm:noether} as well as some conditions under which its
  converse is true. Noether's Theorem \ref{thm:noether} was proved
  previously in \cite{rsv07} using
local coordinates. Here we present a direct global proof using the
Fr\"olicher-Nijenhuis formalism.

\subsection{Conservation laws and  Cartan symmetries}\label{clcs}

  For a regular Lagrangian on $TQ$, the corresponding
  Euler-Lagrange equations are equivalent to a {\sc sode}.
  This implies that its dynamical symmetries are equivalent to Cartan
  symmetries, which (locally) determine and are determined by
  constants of motions \cite[\S 13.8]{crampin}. For $k>1$ and a
  Lagrangian $L$ on $T^1_kQ$ none of the above equivalences are true
  anymore in the very general context. However, some relations remain
  true. In this subsection
  we discuss such relations between Cartan symmetries and conservation laws.

\begin{defn} \label{dynsym}   A map $f=(f^1 ,
    \ldots , f^k)\colon T^1_kQ \to \mathbb{R}^k$ is called a \emph{conservation law} (or a
\emph{conserved quantity}) for the Euler-Lagrange equations
\eqref{eqn:eulerlagrange} if the divergence of $$f\circ\phi^{(1)}=(f^1\circ\phi^{(1)},\ldots,f^k\circ\phi^{(1)})\colon
U\subset \mathbb{R}^k \rightarrow \mathbb{R}^k$$ is zero, for every
$\phi\colon U\subset  R^k\to M$ solution to the Euler-Lagrange equations \eqref{eqn:eulerlagrange}, that is
  \begin{eqnarray}\label{conlaw}
0=\frac{\partial (f^\alpha \circ \phi^{(1)})}{\partial
t^\alpha}\Big\vert_{t}= \frac{\partial f^{\alpha}}{\partial
q^i}\Big\vert_{\phi^{(1)}(t)}\frac{\partial \phi^i}{\partial
t^{\alpha}}\Big\vert_{ t} +\frac{\partial f^{\alpha}} {\partial
v^i_{\beta}}\Big\vert_{ \phi^{(1)}(t)} \frac{\partial^2
\phi^i}{\partial t^{\alpha}
\partial t^{\beta}}\Big\vert_{t}.
\end{eqnarray}
\end{defn}
Now, we present a simple  example of  conservation law.
\begin{ex}\label{exconsla}
The following two functions  $f^\alpha:T^1_2\mathbb{R}\to \mathbb{R}$,
$\alpha \in \{1,2\}$, where
\begin{eqnarray}  f^1(v_1, v_2)=-2\sigma v_1v_2, \quad f^2(v_1,v_2)=\sigma
  (v_1)^2 + \tau (v_2)^2,  \label{noncsym} \end{eqnarray}
give a conservation law for the Euler-Lagrange equation \eqref{vs}.
In fact, if $\phi$ is a solution to the Euler-Lagrange equations
\eqref{vs},  using \eqref{noncsym} we deduce that
  \begin{eqnarray*}
  \frac{\partial (f^1\circ \phi^{(1)})}{\partial t^1} + \frac{\partial (f^2\circ \phi^{(1)})}{\partial t^2}
=  0.
 \end{eqnarray*}
Hence the functions \eqref{noncsym} give a conservation law for
the Lagrangian \eqref{lvs}.
\end{ex}

\begin{lem}\label{imp} Let $f=(f^1,\ldots,f^k)\colon T^1_kQ\rightarrow
\mathbb{R}^k$ be a conservation law. Let $\xi=(\xi_1,\dots,\xi_k)$  be
an integrable {\sc sopde}   in $\mathfrak{X}^k_L(T^1_kQ)$,
  then
\begin{eqnarray} \label{xifa}
\xi_{\alpha}(f^\alpha)=0.
\end{eqnarray}
 \end{lem}
\begin{proof}
 Since $\xi$  is integrable
we know that for every point $x\in T^1_kQ$ there exists an integral
section $\phi^{(1)}:U\subset \mathbb{R}^k  \rightarrow T^1_kQ$ such
that
\begin{enumerate}
\item[1)] $\phi$ is a solution to the Euler-Lagrange equations, because $\xi\in \mathfrak{X}^k_L(T^1_kQ)$,
\item[2)] $\phi$ satisfies
$$ \phi^{(1)}(0)=x, \
d_t\phi^{(1)}\left(\frac{\partial}{\partial
t^\alpha}\Big\vert_t\right)=\xi_{\alpha}(\phi^{(1)} (t))\in
T_{\phi^{(1)}(t)}\left(T^1_kQ\right),
$$
 for every $t\in U$ and for all $\alpha \in \{1, ..., k\}$.
\end{enumerate}
Condition 2) above means that
\begin{equation}\label{caz}
 v^i_\alpha(\phi^{(1)}(t))=\frac{\partial \phi^i} {\partial
t^\alpha}\Big\vert_{t}, \quad  \xi^i_{\alpha
\beta}(\phi^{(1)}(t)) =\frac{\partial^2 \phi^i}{\partial t^{\alpha}
\partial t^{\beta}}\Big\vert_{t}.
\end{equation} Since  $f=(f^1,\ldots,f^k)$ is a conservation
law, using formula \eqref{conlaw} at $t=0$, and using formulae
\eqref{caz}, we have
\begin{eqnarray*}
0 &=& \frac{\partial (f^\alpha \circ \phi^{(1)})}{\partial
t^\alpha}\Big\vert_{0}
 = \frac{\partial
f^{\alpha}}{\partial q^i}\Big\vert_{ \phi^{(1)}(0)} \frac{\partial
\phi^i}{\partial t^{\alpha}}\Big\vert_{ 0} +\frac{\partial
f^{\alpha}} {\partial v^i_{\beta}}\Big\vert_{ \phi^{(1)}(0)}
\frac{\partial^2 \phi^i}{\partial t^{\alpha}
\partial t^{\beta}}\Big\vert_{0} \\ & =  &
 \frac{\partial f^{\alpha}}{\partial q^i}\Big\vert_{ x}
\frac{\partial \phi^i}{\partial t^{\alpha}}\Big\vert_{ 0}
+\frac{\partial f^{\alpha}} {\partial v^i_{\beta}}\Big\vert_{ x}
\frac{\partial^2 \phi^i}{\partial t^{\alpha}
\partial t^{\beta}}\Big\vert_{0}= \frac{\partial f^{\alpha}}{\partial q^i}\Big\vert_{ x}
v^i_\alpha(x) +\frac{\partial f^{\alpha}} {\partial
v^i_{\beta}}\Big\vert_{ x} \xi^i_{\alpha\beta}(x) = \xi_{\alpha}(x)f^\alpha
 \end{eqnarray*}
\end{proof}

The converse of Lemma \ref{imp} may not be true, and the
reason is that,  as we can see from formula \eqref{xiel},  we might
have solutions $\phi$ of the Euler-Lagrange equations \eqref{eqn:eulerlagrange2}  that are not
solutions to some  $\xi \in \mathfrak{X}^k_L(T^1_kQ)$.

However, we will see in the following Lemma   that, under some
assumption on functions $f^{\alpha}$, this converse is true.

\begin{lem} \label{lem:12}
Let $L\in C^{\infty}(T^1_kQ)$ be a Lagrangian and
 assume that there exists  a vector field $X\in
\mathfrak{X}(T^1_kQ)$ such that
\begin{eqnarray}\label{ixoa}  i_X\omega_L^\alpha=df^\alpha, \ \forall \alpha \in
  \{1,...,k\},   \end{eqnarray}
  for some functions  $f^{\alpha}: T^1_kQ \to \mathbb{R}$ .

Then, $f^{\alpha}$ is a conservation law for the Euler-Lagrange equations
\eqref{eqn:eulerlagrange} if and only if
$\xi_{\alpha}(f^{\alpha})=0,$ for all integrable {\sc sopde}   $\xi \in \mathfrak{X}^k_L(T^1_kQ)$.
\end{lem}
\begin{proof}
 The direct implication is given by Lemma
\ref{imp}. Note that for this implication we do not
  need the assumption on the existence of the vector field $X$ that
  satisfies \eqref{ixoa}.

 For the converse implication consider  $X=X^i\partial/\partial q^i + X^i_{\alpha}\partial/\partial
v^i_{\alpha}$ a vector field on $T^1_kQ$  that
  satisfies \eqref{ixoa}.  In view of the second formula \eqref{thetala} we can
write equation \eqref{ixoa} as follows
\begin{eqnarray*}
\left[\left(\frac{\partial^2L}{\partial q^i\partial v^j_{\alpha}} -
    \frac{\partial^2L}{\partial q^j \partial v^i_{\alpha}}\right) X^j
  - g^{\alpha\beta}_{ij} X^j_{\beta}\right]dq^i +
g^{\alpha\beta}_{ij}X^i dv^j_{\beta} = \frac{\partial
  f^{\alpha}}{\partial q^i} dq^i +  \frac{\partial
  f^{\alpha}}{\partial v^j_{\beta}} dv^j_{\beta},
\end{eqnarray*}
and necessarily we have
\begin{eqnarray}
 \frac{\partial f^{\alpha}}{\partial v^j_{\beta}} =
 g^{\alpha\beta}_{ij} X^i. \label{pfg} \end{eqnarray}
Consider now $\phi$ any solution to the Euler-Lagrange equations
\eqref{eqn:eulerlagrange2} (which may not be a solution of any $\xi$). It
follows that $\phi$ satisfies equations \eqref{xiel}, since $\xi$ is
assumed to be an integrable {\sc sopde}.

If we
contract equations \eqref{xiel} by  $X^i\circ \phi^{(1)}$, we obtain
\begin{eqnarray}
(X^i\circ \phi^{(1)}) (g^{\alpha\beta}_{ij}\circ \phi^{(1)}) \left(
  \frac{\partial^2 \phi^j}{\partial t^{\alpha} \partial t^{\beta}} -
  \xi^j_{\alpha\beta}\circ \phi^{(1)}\right) = 0. \label{xgphi}
\end{eqnarray}

If we replace formula \eqref{pfg} in equation \eqref{xgphi} we obtain
\begin{eqnarray*}
0=\frac{\partial f^{\alpha}}{\partial
  v^j_{\beta}}\circ \phi^{(1)}\left(\frac{\partial^2 \phi^j}{\partial
    t^{\alpha}\partial t^{\beta}} - \xi^j_{\alpha\beta} \circ
  \phi^{(1)}\right) = \frac{\partial (f^{\alpha} \circ \phi^{(1)})}{\partial
  t^{\alpha}} - \xi_{\alpha}(f^{\alpha}) \circ \phi^{(1)} \, ,
\end{eqnarray*}
and this formula proves the result.
\end{proof}

Let us recall the definition of Cartan symmetry for a Lagrangian $L$, see \cite{rsv07}.
 \begin{defn}     A vector field $X\in
    \mathfrak{X}(T^1_kQ)$ is called a \emph{Cartan symmetry} of
     the  Lagrangian $L$, if
    $\mathcal{L}_X\omega^{\alpha}_L=0$ for all $\alpha \in
    \{1,...,k\}$ and $\mathcal{L}_XE_L=0$.
\end{defn}
In this case the flow $\phi_t$ of $X$ transforms solutions to the Euler-Lagrange equations on solutions to the Euler-Lagrange equations, that is, each $\phi_t$ is a symmetry of the Euler-Lagrange equations, see \cite{rsv07}.

From condition $\mathcal{L}_X\omega^{\alpha}_L=0$ one obtains that
there exists locally defined functions $f^{\alpha}$ such that
$i_X\omega^{\alpha}_L=df^\alpha$. Thus if $X$ is a Cartan symmetry
Lemma \ref{lem:12} holds locally.

\subsection{Newtonoid vector fields}\label{subnewt} In this subsection we study some
properties of the set of Newtonoid vector fields associated to a {\sc sopde},
generalizing the $k=1$ case, introduced by Marmo and Mukunda in
\cite{marmo86}.  Properties of the
  newtonoid vector fields associated to a {\sc sode} and their
  relations to symmetries and first order variation of geodesics were studied \cite{bcd11}.
We extend some of these properties to the case $k>1$. In Proposition
\ref{prop:csn} we prove that Cartan symmetries of a regular  Lagrangian
$L$ are Newtonoid vector fields for all corresponding {\sc sopde} $\xi \in
\mathfrak{X}^k_L(T^1_kQ)$.

We fix a {\sc sopde} $\xi$ and consider the following set of vector fields on $T^1_kQ$
\begin{eqnarray}
\mathfrak{X}_{\xi}=\operatorname{Ker}\left(J^{\alpha}\circ
\mathcal{L}_{\xi_{\alpha}}\right)\subset \mathfrak{X}(T^1_kQ). \label{xs1}
\end{eqnarray}
The set $\mathfrak{X}_{\xi}$ can be expressed locally as follows
\begin{eqnarray}
\mathfrak{X}_{\xi}=\left\{X\in \mathfrak{X}(T^1_kQ), \
X=X^i\frac{\partial}{\partial q^i} +
\xi_{\alpha}(X^i)\frac{\partial}{\partial v^i_{\alpha}}\right\}. \label{xs2}
\end{eqnarray}
Indeed, for a vector field $X\in \mathfrak{X}(T^1_kQ)$,
we have
\begin{eqnarray} \nonumber \left[\xi_\alpha, X\right] &=&
  \left[v_\alpha^i \frac{\partial}{\partial q^i} +
\xi^i_{\alpha\beta} \frac{\partial}{\partial v^i_{\beta}}, X^i \frac{\partial}{\partial q^i} +
X^i_{\alpha} \frac{\partial}{\partial v^i_{\alpha}}\right ] \\
&=& \left( \xi_\alpha(X^i) -X^i_{\alpha}\right)
\frac{\partial}{\partial q^i} + \left( \xi_\alpha(X^i_{\beta})
  -X(\xi^i_{\alpha\beta})\right)\frac{\partial}{\partial v^i_\beta}, \label{xiax}\end{eqnarray}
and therefore $J^{\alpha}\circ
\mathcal{L}_{\xi_{\alpha}}(X)=0$ if and only if  $\xi_\alpha(X^i) =X^i_\alpha.$
\begin{defn} Consider $\xi$ a {\sc sopde}. \begin{itemize}  \item[1)]  A vector field $X\in
      \mathfrak{X}_{\xi}$ is called a \emph{Newtonoid}  for
      $\xi$. \item[2)] A vector field $X\in \mathfrak{X}(T^1_kQ)$ is
      called a \emph{dynamical symmetry} of $\xi$ if $[\xi_{\alpha}, X]=0$, for all $\alpha \in
    \{1,...,k\}$. \end{itemize}
\end{defn}

From formula \eqref{xiax} it  follows that any dynamical symmetry for a
{\sc sopde} $\xi$ is a Newtonoid for $\xi$. For $k=1$, the set
$\mathfrak{X}_{\xi}$ was introduced in \cite{marmo86}
and it was called the set of Newtonoid vector fields.
In the next lemma we provide some
properties of the set of Newtonoid vector fields.

\begin{lem} \label{lem:pis}
For a {\sc sopde} $\xi $ consider the map $\pi_{\xi}:
\mathfrak{X}(T^1_kQ) \to \mathfrak{X}(T^1_kQ)$, given by  $\pi_{\xi} =
\operatorname{Id} + J^{\alpha}\circ \mathcal{L}_{\xi_{\alpha}}$. \begin{itemize}
\item[1)] The map $\pi_{\xi}$  satisfies $\pi_{\xi}\circ \pi_{\xi'}=\pi_{\xi}$,  for
  any two {\sc sopde}s $\xi$ and $\xi'$. In particular we have $\pi_{\xi}^2=\pi_{\xi}$
  and hence $\pi_{\xi}$ is
  a projector; \item[2)] $\operatorname{Im} \pi_{\xi}=\mathfrak{X}_{\xi}$,
  $ \operatorname{Ker} \pi_{\xi} = \mathfrak{X}^v(T^1_kQ)$, and hence the
  following sequence is exact $$ 0 \to \mathfrak{X}^v(T^1_kQ) \stackrel{i}{\to} \mathfrak{X}(T^1_kQ) \stackrel{\pi_{\xi}}{\to}
  \mathfrak{X}_{\xi} \to 0.$$ \item[3)] For $f\in C^{\infty}(T^1_kQ)$ and $X\in
\mathfrak{X}_{\xi}$, we define the product
\begin{eqnarray} f\ast X=\pi_{\xi}(fX) =
  fX+\xi_{\alpha}(f)J^{\alpha}X \in
\mathfrak{X}_{\xi}. \label{star} \end{eqnarray} The set
$\mathfrak{X}_{\xi}$ is a $C^{\infty}(T^1_kQ)$-module with respect to the $\ast$ product.
\item[4)] A vector field $X$ on $T^1_kQ$ is a Newtonoid
  for $\xi$ if and only if it has the local expression
\begin{eqnarray} X=X^i(q,v)\ast \frac{\partial}{\partial
    q^i}. \label{xstar} \end{eqnarray}
\end{itemize}
\end{lem}
\begin{proof}
Using the definition of the map $\pi_{\xi}$ it follows that
$\pi_{\xi}\circ \pi_{\xi'}=\operatorname{Id} + J^{\alpha}\circ
\mathcal{L}_{\xi_{\alpha}} + J^{\alpha}\circ
\mathcal{L}_{\xi'_{\alpha}}  +  J^{\alpha}\circ
\mathcal{L}_{\xi_{\alpha}} \circ J^{\beta}\circ
\mathcal{L}_{\xi'_{\beta}}.$ Now using the formula $ J^{\alpha}\circ
\mathcal{L}_{\xi_{\alpha}} \circ J^{\beta}=-J^{\beta}$ it follows that $\pi_{\xi}\circ \pi_{\xi'}=\pi_{\xi}$, which shows that
first part of the lemma is true.

Using formula \eqref{xiax}, we have
$$ \pi_{\xi} \left( X^i\frac{\partial}{\partial q^i} +
X^i_{\alpha}\frac{\partial}{\partial v^i_{\alpha}}\right) = X^i\frac{\partial}{\partial q^i} +
\xi_{\alpha}(X^i)\frac{\partial}{\partial v^i_{\alpha}},$$  which shows
that $\operatorname{Im} \pi_{\xi}=\mathfrak{X}_{\xi}$ and
$\operatorname{Ker} \pi_{\xi} = \mathfrak{X}^v(T^1_kQ)$.

With the $\ast$ product defined in formula \eqref{star}, the map
$\pi_{\xi}$ transfers the $C^{\infty}(T^1_kQ)$-module structure of
$\left(\mathfrak{X}(T^1_kQ), \cdot\right)$ to $\left(\mathfrak{X}_{\xi},
  \ast\right)$.

We have that $\operatorname{Im} \pi_{\xi}=\mathfrak{X}_{\xi}$. Therefore
$X\in \mathfrak{X}_{\xi}$ if and only if $X=\pi_{\xi}(X)$. Using the above
properties of map  $\pi_{\xi}$, a vector field $X$ on $T^1_kQ$ is locally
given as in the last part of formula  \eqref{xs2} if and only if it is
given by formula \eqref{xstar}.
\end{proof}

From formula \eqref{clift} it follows that the complete lift $Z^c\in
\mathfrak{X}(T^1_kQ)$ of a vector field $Z\in \mathfrak{X}(Q)$ is a
Newtonoid vector field for an arbitrary {\sc sopde} $\xi$. In the next
proposition we will see that the set of Newtonoid vector fields
contains also Cartan symmetries.

\begin{prop} \label{prop:csn}
Consider $L$ a regular Lagrangian on $T^1_kQ$ and $X\in \mathfrak{X}(T^1_kQ)$ a
Cartan symmetry of $L$. Then $X$ is a Newtonoid vector field  for every $\xi \in
\mathfrak{X}^k_L(T^1_kQ)$. \end{prop}
\begin{proof}
Consider $X\in \mathfrak{X}(T^1_kQ)$ a Cartan symmetry of $L$ and $\xi
\in \mathfrak{X}^k_L(T^1_kQ)$. Since $L$ is regular it follows that
$\xi$ is a {\sc sopde}. Moreover, $\xi$ is a solution of the equation
$i_{\xi_{\alpha}}\omega^{\alpha}_L=dE_L$.  If we apply $\mathcal{L}_X$ to
both sides of this equation and use the commutation
  rules we obtain
\begin{eqnarray*} i_{\xi_{\alpha}}\mathcal{L}_X \omega^{\alpha}_L -
  i_{[\xi_{\alpha}, X]}\omega^{\alpha}_L=d\mathcal{L}_X E_L. \end{eqnarray*}
Using now the fact $\mathcal{L}_X \omega^{\alpha}_L=0$ and $\mathcal{L}_X E_L=0$ it
follows that
\begin{eqnarray}  i_{[\xi_{\alpha},
    X]}\omega^{\alpha}_L=0. \label{isaxoal}\end{eqnarray}
We will prove now that equation \eqref{isaxoal} implies that
$J^{\alpha}[\xi_{\alpha}, X]=0$ and hence $X$ is a Newtonoid vector
field for $\xi$.
Using formula \eqref{xiax}, we have
\begin{eqnarray} [\xi_{\alpha}, X]=V_{\alpha}^i\frac{\partial}{\partial q^i} +
V^i_{\alpha\beta} \frac{\partial}{\partial v^i_{\beta}}, \label{sax} \end{eqnarray}
where $V^i_{\alpha}=\xi_{\alpha}(X^i)-X^i_{\alpha}$ and
$V^i_{\alpha\beta}=\xi_{\alpha}(X^i_{\beta})-
X(\xi^i_{\alpha\beta})$.
Using formula \eqref{thetala}, it follows that the $k$-symplectic
$2$-forms $\omega^{\alpha}_L$ can be written as follows
\begin{eqnarray}
\omega^{\alpha}_L=a^{\alpha}_{ij} dq^i\wedge dq^j +
  g^{\alpha\beta}_{ij} dq^i \wedge dv^j_{\beta}, \label{oal} \end{eqnarray} where
$$ a^{\alpha}_{ij}= \frac{1}{2}\left( \frac{\partial^2 L}{\partial
    q^j \partial v^i_{\alpha}} - \frac{\partial^2 L}{\partial
    q^i \partial v^j_{\alpha}} \right), \quad g^{\alpha\beta}_{ij} =
\frac{\partial^2 L}{\partial v^i_{\alpha} \partial
  v^j_{\beta}}. $$
If we replace now formulae \eqref{sax} and \eqref{oal} in equation
\eqref{isaxoal} we obtain
$$ \left(2a^{\alpha}_{ij} V^j_{\alpha} -
  g^{\alpha\beta}_{ij}V^j_{\alpha\beta}\right) dq^i + g^{\alpha
  \beta}_{ij} V^i_{\alpha} dv^j_{\beta}=0, $$
which implies that $g^{\alpha \beta}_{ij} V^i_{\alpha}=0$. Using the fact that the Lagrangian $L$
is regular it follows that $g^{\alpha \beta}_{ij} $ has maximal rank
and hence $V^i_{\alpha}=\xi_{\alpha}(X^i)-X^i_{\alpha}=0$, which shows
that $X$ is a Newtonoid vector field for $\xi$.
\end{proof}

\subsection{Noether's Theorem} \label{sec:noether}

For the $k=1$ case it is well known that Cartan symmetries induce and
are induced by constants of motion, and these results are known as
Noether's Theorem and its converse. For $k>1$, Noether's Theorem is also
true, each Cartan symmetry induces a conservation law, see Theorem
\ref{thm:noether}. However, its converse may not be true. In
Proposition \ref{convnt} we discuss when this is the case.

The following theorem is proved in \cite[Thm 3.13]{rsv07} using local
coordinates. Here we give a direct proof of Noether's Theorem using the
Fr\"olicker-Nijenhuis formalism on
$T^1_kQ$. This proof will allow us to discuss also when the
converse of Noether's Theorem is true, for $k>1$. To show that there are
cases when the converse of Noether's Theorem is not true, we
  provide examples of conservation laws that are not induced by Cartan
  symmetries.

\begin{thm} (Noether's Theorem) \label{thm:noether} Consider  $L$ a
  Lagrangian on $T^1_kQ$ and $X\in \mathfrak{X}(T^1_kQ)$
  a Cartan symmetry for $L$.  Then, there  exists (locally defined) functions
  $g^{\alpha}$ on $T^1_kQ$ such that
  \begin{eqnarray}
    \mathcal{L}_X\theta^\alpha_L=dg^{\alpha} \label{lxdja} \end{eqnarray} and the following
  functions
\begin{eqnarray} f^{\alpha}=\theta^\alpha_L(X)-g^\alpha \label{conslaw}\end{eqnarray}
give a conservation law for the Euler-Lagrange equations.
\end{thm}
 \begin{proof}
Since $X$ is a Cartan symmetry for $L$ it follows that
$\mathcal{L}_X\omega^{\alpha}_L=0$ and hence
the $1$-forms
$\mathcal{L}_X\theta^\alpha_L$ are closed. Locally, on $T^1_kQ$, one can find
function $g^{\alpha}$ such that
$\mathcal{L}_X\theta^{\alpha}_L=dg^{\alpha}$, thus
$$
i_Xd\theta^{\alpha}_L+di_X\theta^{\alpha}_L=dg^{\alpha},
$$
or equivalently
$$
i_X \omega_L^{\alpha}=d(\theta^\alpha_L(X)-g^\alpha).
$$
We will show now that functions
$f^{\alpha}=\theta^\alpha_L(X)-g^\alpha$, in formula \eqref{conslaw},
give a conservation law.  We will compute first $\xi_\alpha( f^{\alpha})$.

Using formula \eqref{conslaw} we have \begin{eqnarray} \nonumber \xi_{\alpha}(f^{\alpha}) &=& \mathcal{L}_{\xi_{\alpha}}
i_X\theta^{\alpha}_L-\xi_{\alpha}(g^{\alpha})  =
i_X\mathcal{L}_{\xi_{\alpha}}d_{J^{\alpha}}L + i_{[\xi_{\alpha}, X]}
d_{J^{\alpha}}L - \xi_{\alpha}(g^{\alpha}) \\
&=&   i_XdL +  i_{[\xi_{\alpha}, X]}
d_{J^{\alpha}}L  - \xi_{\alpha}(g^{\alpha}).  \label{safa1}\end{eqnarray}

We apply now $i_{\xi_{\alpha}}$ to both terms in formula
  $\mathcal{L}_Xd_{J^{\alpha}}L=dg^{\alpha}$, sum over $\alpha$, and
  obtain
\begin{eqnarray*}
\nonumber \xi_{\alpha}(g^{\alpha}) &=& i_{\xi_{\alpha}}dg^{\alpha} =i_{\xi_{\alpha}}
\mathcal{L}_Xd_{J^{\alpha}}L =
\mathcal{L}_X  i_{\xi_{\alpha}} d_{J^{\alpha}}L + i_{[\xi_{\alpha},
  X]}d_{J^{\alpha}}L \\
&=&  \mathcal{L}_X \mathbb{C}(L) +
 i_{[\xi_{\alpha}, X]} d_{J^{\alpha}}L .  \label{safa2}
\end{eqnarray*}
If we replace now, $\xi_{\alpha}(g^{\alpha})$, from
  the above formula in formula \eqref{safa1} we obtain $\xi_{\alpha}(f^{\alpha})
=-\mathcal{L}_X(E_L)=0$. Therefore we have:
$$\xi_{\alpha}(f^{\alpha})=0, \quad i_X\omega^{\alpha}_L=df^{\alpha}.$$
Now using Lemma \ref{lem:12} it follows that $f^{\alpha}$ is
a conservation law for $L$.
\end{proof}

We have seen that if $X$ is a Cartan symmetry for a
  Lagrangian $L$ on $T^1_kQ$ then the functions $f^{\alpha}\in
  C^{\infty}(T^1_kQ)$, which satisfy the equation
  $i_X\omega^{\alpha}_L=df^{\alpha}$, give a conservation law for
  $L$. We say that this conservation law $f^{\alpha}$ is induced by the Cartan
  symmetry $X$. For $k>1$ there are conservation laws that are not
  induced by Cartan symmetries. Next we provide such an example.

\begin{ex}\  \label{ex3}

\noindent $a)$
We have seen in Example \ref{exconsla} that the functions
$f^\alpha:T^1_2\mathbb{R}\to \mathbb{R}$, given by formula \eqref{noncsym},
give a conservation law for the Euler-Lagrange equations
\eqref{vs}. We will prove now that this conservation law is not
induced by a Cartan symmetry, and hence it will
  show that the converse of Noether's Theorem
\ref{thm:noether} is not true, unless the assumptions
\eqref{ixoa}  are satisfied.
 Consider $X\in \mathfrak{X}(T^1_2\mathbb{R})$, locally given by
$$
X= Z\frac{\partial}{\partial q}+Z_{1} \frac{\partial}{\partial v_{1}}
+Z_{2} \frac{\partial}{\partial v_{2}}
$$
Using formulae \eqref{locomel}, first equation
  \eqref{ixoa}, for $\alpha=1$, can be written as follows

$$i_X\omega^1_L=\sigma(Z dv_1 - Z_1 dq)=df^1=\frac{\partial f^1}{\partial q} dq+ \frac{\partial f^1}{\partial v_1} dv_1+\frac{\partial f^1}{\partial v_2} dv_2.$$

This implies that $\frac{\partial f^1}{\partial
    v_2}=0$, which is not true, since in our case $\frac{\partial
  f^1}{\partial v_2}=-2\sigma v_1$.

\noindent  $b)$ Consider the homogeneous isotropic 2-dimensional wave equation
\begin{eqnarray} u_{tt}-cu_{xx}-cu_{yy}=0.\label{we} \end{eqnarray}
Let us make the following notations
$t^1=t, t^2=x, t^3=y$ and $q=u$. The regular Lagrangian function $L\in
C^{\infty}(T^1_3\mathbb{R})$ for the wave
equation \eqref{we} is
\begin{eqnarray}
  L=\frac{1}{2}\left((v_1)^2-c(v_2)^2- c(v_3)^2\right). \label{lwe}\end{eqnarray}
Each of the following three sets of functions on $T^1_3\mathbb{R}$ will give a
conservation law for the Lagrangian $L$ in formula \eqref{lwe}:
$$\begin{array}{l}
\nonumber f^1(v_1,v_2,v_3) = (v_1)^2+c(v_2)^2+c(v_3)^2, \  f^2(v_1, v_2,
v_3)=-2cv_1v_2, \ f^3(v_1, v_2, v_3)=-2cv_1v_3; \\ \noalign{\medskip}
f^1(v_1, v_2, v_3)=2v_1v_2, \ f^2(v_1,v_2,v_3) =
-(v_1)^2-c(v_2)^2+c(v_3)^2, \  f^3(v_1, v_2, v_3)=-2cv_2v_3;  \\ \noalign{\medskip}
f^1(v_1, v_2, v_3)=2v_1v_3, \ f^2(v_1, v_2, v_3)=-2cv_2v_3, \  f^3(v_1,v_2,v_3) =
-(v_1)^2+c(v_2)^2-c(v_3)^2.
\end{array}$$
None of these conservation laws are induced by Cartan symmetries. \end{ex}

Theorem \ref{thm:noether} shows that any Cartan symmetry of  a
 Lagrangian  $L$ induces (locally defined) conservation
laws, for $k\geq 1$.

For the case $k=1$ the converse of this theorem
is also true: any conservation law of a
 Lagrangian  is induced by a Cartan symmetry.

In the case $k>1$ such result is not true
anymore, unless we require some extra assumptions.
 As we have already seen in Example \ref{ex3} there
  are examples of conservation
laws for some Lagrangians,  that are not induced by
any Cartan symmetries.

Part of the next proposition will show when conservation laws for a
 Lagrangian  are induced by Cartan symmetries.

\begin{prop} \label{convnt}
Consider $L\in C^{\infty}(T^1_kQ)$ a Lagrangian, functions
$f^1,  \ldots , f^k\in  C^{\infty}(T^1_kQ)$,  and a vector field $X\in
\mathfrak{X}(T^1_kQ)$ such that equations  \eqref{ixoa} are satisfied.
Then $f^{\alpha}$ is a conservation law for $L$ if and only if
 $X$ is Cartan symmetry.
\end{prop}
\begin{proof}
In view of Lemma \ref{lem:12} we will have to prove that $X$ is a
Cartan symmetry if and only if $\xi_{\alpha}(f^{\alpha})=0$ for all
integrable {\sc sopde}  $\xi \in \mathfrak{X}^k_L(T^1_kQ)$.

Using formula \eqref{ixoa},  and the fact that $\xi \in \mathfrak{X}^k_L(T^1_kQ)$, we have
\begin{eqnarray}
X(E_L)=i_XdE_L=-i_Xi_{\xi_{\alpha}}\omega^{\alpha}_L=i_{\xi_{\alpha}}i_X\omega^{\alpha}_L=i_{\xi_{\alpha}}
df^{\alpha}=\xi_{\alpha}(f^{\alpha}). \label{xel}
\end{eqnarray}
From formula \eqref{ixoa} it follows that
$\mathcal{L}_X\omega^{\alpha}_L=0$. Therefore, $X$ is a Cartan
symmetry if and only if $X(E_L)=0$ and, in view of formula
\eqref{xel}, this is equivalent to $\xi_{\alpha}(f^{\alpha})=0$.
\end{proof}

For the case $k=1$, the regularity condition of the Lagrangian implies
that the Poincar\'e-Cartan $2$-form $\omega$ is a symplectic form and
hence equation \eqref{ixoa} always has a unique solution. In the case
$k>1$, for some given functions $f^{\alpha}\in C^{\infty}(T^1_kQ)$,
the system \eqref{ixoa} is overdetermined and it may not have
solutions $X\in \mathfrak{X}(T^1_kQ)$. We will provide examples when
this is the case.

\begin{ex}
  \begin{enumerate}

  \item[1)]
  Let us consider the following Lagrangians $L:T^1_2\mathbb{R} \to \mathbb{R}$:
  $$(a) \quad L(q,v_1,v_2)=\frac{\displaystyle 1}{\displaystyle 2}(\sigma
v^2_1-\tau v^2_2), \quad (b) \quad L(q,v_1,v_2)=\sqrt{1+(v_1)^2+(v_2)^2}.$$
The vector field   $X={\partial}/{\partial q}$ on $T^1_kQ$ is a Cartan
symmetry for both Lagrangians and the corresponding conservation laws  are
$$(a) \quad f^1=\sigma v_1,
  \quad f^2=-\tau v_2,\quad  (b) \quad f^1=\frac{v_1}{\sqrt{1+(v_1)^2+(v_2)^2}}, \quad f^2=\frac{v_2}{\sqrt{1+(v_1)^2+(v_2)^2}}.$$

  The above Lagrangians correspond to the vibrating string equations
  and the equation of minimal surfaces, respectively, see \cite{Mart2,Olver}.

 \item[2)]
  For the  Lagrangian $L : T^1_3\mathbb{R} \to \mathbb{R}$
defined by $$L(q,v_1,v_2,v_3)=\frac{1}{2}((v_1)^2+(v_2)^2+(v_3)^2),$$
 the vector field  $X={\partial}/{\partial q}$ is a Cartan symmetry, and the induced conservation law is
$$
f^1(v_1)=v_1, \quad f^2(v_2)=v_2, \quad f^3(v_3)=v_3.
$$
The Euler-Lagrange equations corresponding to $L$ are the Laplace equations.

\item[3)]  For the  Lagrangian $L: T^1_2 \mathbb{R}^2 \to  \mathbb{R}$ defined by
 $$
 L(q^1,q^2,v^1_1,v^1_2,v^2_1,v^2_2) = \left(\frac{1}{2}\lambda +\nu\right)
 [(v^1_1)^2+(v^2_2)^2] + \frac{1}{2}\nu[(v^1_2)^2+(v^2_1)^2]+(\lambda +\nu)v^1_1v^2_2,$$
the vector field $X= \partial/\partial q^1\, +\, \partial/\partial
q^2$  is again  a Cartan symmetry. The induced conservation law is
$$
f^1=(\lambda +2\nu)v^1_1+\nu v^2_1+(\lambda + \nu)v^2_2, \quad
f^2= (\lambda +\nu)v^1_1+\nu v^1_2+(\lambda + 2\nu) v^2_2   .
$$
 The Euler-Lagrange equations corresponding to $L$ are the Navier equations, see \cite{Mart2,Olver}
\end{enumerate}
\end{ex}

In Proposition \ref{prop:csn} we have seen that Cartan symmetries are
Newtonoid vector fields. Next theorem shows that under some
assumptions Newtonoid vector fields provide Cartan symmetries and
hence conservation laws. This theorem generalizes the result obtained
in the case $k=1$ by Marmo and Mukunda \cite{marmo86} for regular Lagrangians.

\begin{thm} \label{thm:cscl} Consider $L$ a   regular Lagrangian
    on $T^1_kQ$. We assume that there exists $X\in \mathfrak{X}(T^1_kQ)$ and
  $g^{\alpha}\in C^{\infty}(T^1_kQ)$ such that
\begin{eqnarray} \pi_{\xi}(X)(L)=\xi_{\alpha}(g^{\alpha}), \forall
  {\ \textrm  {\sc sopde} \ } \xi_{\alpha}. \label{marmo} \end{eqnarray}
Then, it follows: \begin{itemize} \item[1)] If $\xi\in
  \mathfrak{X}^k_L(T^1_kQ)$ we have that $\pi_{\xi}(X)$ is a Cartan
  symmetry for $L$. \item[2)] The functions $f^{\alpha}=\theta_L^\alpha(X)-g^{\alpha} $ give a
  conservation law for $L$. \end{itemize}
\end{thm}
\begin{proof}

$1)$ We have to prove that the following two conditions are satisfied
$$(a)\quad \mathcal{L}_{\pi_{\xi}(X)} \omega^{\alpha}_L=0, \quad (b)
\quad  \mathcal{L}_{\pi_{\xi}(X)} E_L=0.$$

$(a)$  For each $\alpha \in \{1,...,k\}$ we denote the $1$-forms
\begin{eqnarray}
\eta^{\alpha}=\mathcal{L}_{\pi_{\xi}(X)}\theta^\alpha_L-dg^{\alpha}. \label{eta} \end{eqnarray}
First we show that
$$
i_{V_\alpha}\eta^\alpha=0, \quad L_{\xi_\alpha}\eta^\alpha(V)=0
$$
for arbitrary vertical vector fields
$V, V_1, V_2, \ldots , V_k.$  First condition above is equivalent to
the fact  that $\eta^{\alpha}=\eta^{\alpha}_idq^i$ are semi-basic
$1$-forms. Moreover, using the fact that
$\mathcal{L}_{\xi_{\alpha}}\eta^{\alpha}=\xi_{\alpha}(\eta^{\alpha}_i)
dq^i + \eta^{\alpha}_i dv^i_{\alpha} $,  it follows that the second condition  will imply
$\eta^{\alpha}=0.$

Let  $\xi_{\alpha}\in \mathfrak{X}(T^1_kQ)$ be a {\sc sopde} and
$V_{\alpha}$ be vertical vector fields. It follows that $\xi'_{\alpha}=\xi_{\alpha} -
V_{\alpha}$ are also {\sc sopde}s. Using the fact that
$\theta^\alpha_L$ are semi-basic $1$-forms and $V_{\alpha}$ are
vertical vector  fields, we have that $i_{V_{\alpha}}
\theta^\alpha_L=0$. Therefore, making use of the corresponding
commutation rules, we have
\begin{eqnarray*}
i_{V_{\alpha}}\eta^{\alpha} & = & i_{V_{\alpha}}
\mathcal{L}_{\pi_{\xi}(X)}\theta^\alpha_L- i_{V_{\alpha}}  dg^{\alpha} = i_{V_{\alpha}}
\mathcal{L}_{\pi_{\xi}(X)}\theta^\alpha_L -
\mathcal{L}_{\pi_{\xi}(X)} i_{V_{\alpha}} \theta^\alpha_L-
i_{V_{\alpha}}  dg^{\alpha}  \\
& =&  i_{[V_{\alpha}, \pi_{\xi}(X)]} \theta^\alpha_L - i_{V_{\alpha}}
dg^{\alpha}  =  i_{J^{\alpha}[V_{\alpha}, \pi_{\xi}(X)]} dL - i_{V_{\alpha}}
dg^{\alpha} \\ & =&  i_{\pi_{\xi}(X)}dL -i_{\pi_{\xi'}(X)}dL -  i_{V_{\alpha}}
dg^{\alpha} = \xi_{\alpha}(g^{\alpha}) - \xi'_{\alpha}(g^{\alpha}) -
V_{\alpha}(g^{\alpha}) = 0.
\end{eqnarray*}
In the above calculations we did use the fact that
$\pi_{\xi}(X)-\pi_{\xi'}(X)=J^{\alpha}[V_{\alpha}, \pi_{\xi}(X)]$ and the fact
that the {\sc sopde}s $\xi$ and $\xi'$ satisfy the hypothesis
\eqref{marmo}.

We fix now $\xi\in \mathfrak{X}^k_L(T^1_kQ)$, and since $L$ is regular
this means that $\mathcal{L}_{\xi_{\alpha}}\theta^\alpha_L=dL$. Using
the notation \eqref{eta} we have
\begin{eqnarray*}
\mathcal{L}_{\xi_{\alpha}}\eta^{\alpha} & = & \mathcal{L}_{\xi_{\alpha}}
\mathcal{L}_{\pi_{\xi}(X)}\theta^\alpha_L-\mathcal{L}_{\xi_{\alpha}}
dg^{\alpha} = \mathcal{L}_{[\xi_{\alpha}, \pi_{\xi}(X)]}
\theta^\alpha_L + \mathcal{L}_{\pi_{\xi}(X)}  \mathcal{L}_{\xi_{\alpha}}
\theta^\alpha_L - \mathcal{L}_{\xi_{\alpha}} dg^{\alpha} \\ &=&
i_{[\xi_{\alpha}, \pi_{\xi}(X)]}\omega^{\alpha}_L + \mathcal{L}_{\pi_{\xi}(X)}
dL - d \xi_{\alpha}(g^{\alpha}) = i_{[\xi_{\alpha}, \pi_{\xi}(X)]}\omega^{\alpha}_L.
\end{eqnarray*}
Using the fact that $[\xi_{\alpha}, \pi_{\xi}(X)] $ is a vertical vector
fields, and the $k$-symplectic structure in formula \eqref{thetala}
vanishes on pairs of vertical vector fields, it follows that for an
arbitrary vertical vector field $V$ we have
$$ \mathcal{L}_{\xi_{\alpha}}\eta^{\alpha}(V) =
\omega^{\alpha}_L([\xi_{\alpha}, \pi_{\xi}(X)], V)=0. $$
Hence, we proved that $\eta^{\alpha}=0$, which means that
\begin{eqnarray}
\mathcal{L}_{\pi_{\xi}(X)}\theta^\alpha_L = dg^{\alpha}. \label{dga} \end{eqnarray}
 If we take the exterior derivative in the above formula it follows
 that $\mathcal{L}_{\pi_{\xi}(X)} \omega^{\alpha}_L=0$.

 $(b)$ In order to
 prove that $\pi_{\xi}(X)$ is a Cartan symmetry it remains to show that
 $\pi_{\xi}(X)(E_L)=0.$ For this we use the fact that
 $\mathbb{C}=J^{\alpha}(\xi_{\alpha})$ and hence
 $\mathbb{C}(L)=i_{\mathbb{C}}dL=i_{\xi_{\alpha}}\theta^\alpha_L$. Therefore,
\begin{eqnarray*}
\pi_{\xi}(X)(E_L) & =& \pi_{\xi}(X)(\mathbb{C}(L))-\pi_{\xi}(X)(L) =
\mathcal{L}_{\pi_{\xi}(X)} i_{\xi_{\alpha}}\theta^\alpha_L -
\mathcal{L}_{\pi_{\xi}(X)}L \\ &=&  i_{\xi_{\alpha}} \left(
  \mathcal{L}_{\pi_{\xi}(X)}\theta^\alpha_L - dg^{\alpha} \right) =
0. \end{eqnarray*}
$2)$ So far we have proved that $\pi_{\xi}(X)$ is a Cartan symmetry and it
satisfies formula \eqref{dga}. Using the fact $J^{\alpha} \circ
  \pi_{\xi}=J^{\alpha}$ and Noether's theorem \ref{thm:noether}, it
  follows that the functions $f^{\alpha}=\theta^\alpha_L(\pi_{\xi}(X))-g^{\alpha} = \theta^\alpha_L(X)-g^{\alpha}$ give a conservation
law for $L$.
\end{proof}

Theorem \ref{thm:cscl} extends the results in Corollary 3.15
  from \cite{rsv07}. Indeed if $X=Z^C$ for some $Z\in
\mathfrak{X}(M)$ and $g^{\alpha}\in C^{\infty}(M)$ the condition
\eqref{marmo} becomes $Z^c(L)=v^i_{\alpha}\partial g^{\alpha}/\partial
q^i$. It follows that $Z^c$ is a Cartan symmetry and the functions
$f^{\alpha}=Z^{v_{\alpha}}(L)-g^{\alpha}$ define a
conservation law.

\subsection*{Acknowledgments}
The work of IB was supported by the Romanian National
Authority for Scientific Research, CNCS UEFISCDI, project number
PN-II-ID-PCE-2012-4-0131.

We acknowledge the financial support of the Ministerio de Econom\'{\i}a y Competitividad (Spain), projects MTM2011-22585 and MTM2011-15725-E.

We express our thanks to the referees for their comments and suggestions.


\begin{thebibliography}{MM}

\bibitem{am}
R.A. Abraham, J.E. Marsden. Foundations of Mechanics,
(Second Edition),  Benjamin-Cummings Publishing Company, New York, (1978).

\bibitem{arn}
V.I. Arnold.
Mathematical methods of classical mechanics,
 {\sl Graduate Texts in Mathematics} {\bf 60}. Springer-Verlag, New York-Heidelberg, (1978).

\bibitem{aw1}
 A. Awane. $k$-symplectic structures,
{\sl J. Math. Phys.} {\bf 33} (1992), 4046-4052.

\bibitem{aw3} A. Awane, M. Goze.
Pfaffian systems, $k$-symplectic systems,
Kluwer Academic Publishers, Dordrecht (2000).

\bibitem{BSF88} E. Binz, J. Sniatycki, H. Fischer. Geometry of
  classical fields, North-Holland Mathematics Studies, 154(1988).

 \bibitem{bcd11}
 I. Bucataru, O.A. Constantinescu, M.F. Dahl.
 A geometric setting for systems of ordinary differential equations,
  {\sl International Journal of Geometric Methods in Modern   Physics}, {\bf 8} (6) (2011), 1292--1327.

 \bibitem{bd10} I. Bucataru, M.F. Dahl.  A complete lift for
   semisprays, {\sl International Journal of Geometric Methods in
     Modern Physics}, {\bf 7}(2) (2010), 267--287.

\bibitem{Cant1}
F. Cantrijn, A. Ibort, M. de Le\'on.
On the geometry of multisymplectic manifolds,
{\sl J. Austral. Math. Soc. Ser. A} {\bf 66} (1999), 303-330.

\bibitem{Cant2}
F. Cantrijn,  A. Ibort, M. de Le\'on.
Hamiltonian structures on multisymplectic manifolds,
 {\sl Rend. Sem. Mat. Univ. Politec. Torino}, {\bf 54} (1996), 225-236.

\bibitem{crampin} M. Crampin, F.A.E. Pirani. Applicable differential
  geometry. Cambridge University Press, 1986.

\bibitem{EMR-99}
A. Echeverr\'\i a-Enr\'\i quez, M.C. Mu\~noz-Lecanda, N. Rom\'an-Roy.
Multivector Field Formulation of Hamiltonian Field Theories:  Equations and Symmetries,
{\sl J. Phys. A: Math. Gen.} {\bf 32}(48) (1999) 8461-8484.

\bibitem{GP1}
P.L. Garc\'{\i}a, A. P\'{e}rez-Rend\'{o}n.
Symplectic approach to the theory of quantized fields, I,
{\sl Comm. Math. Phys.} {\bf 13} (1969) 24-44.

\bibitem{GP2}
P.L.  Garc\'{\i}a,  A. P\'{e}rez-Rend\'{o}n.
Symplectic approach to the theory of quantized fields, II,
{\sl Arch. Ratio. Mech. Anal.} {\bf 43} (1971), 101-124.

\bibitem{Sarda2}
G. Giachetta, L. Mangiarotti, G. Sardanashvily.
New Lagrangian and Hamiltonian Methods in Field Theory,
{\sl World Scientific Pub. Co} , Singapore (1997).

\bibitem{GS}
 H. Goldschmidt, S. Sternberg.
The Hamilton-Cartan formalism in the calculus of variations,
{\sl Ann. Inst. Fourier} {\bf 23} (1973), 203-267.

\bibitem{Go1} M.J. Gotay.
An exterior differential systems approach to the Cartan form,
{\sl Symplectic geometry and mathematical physics.
(Aix-en-Provence, 1990)}. Progr. Math., 99, Birkh\"{a}user Boston,
Boston, MA, 1991, pp. 160-188.

\bibitem{Go2}
 M.J. Gotay.
A multisymplectic framework for classical field theory and the calculus of variations, I. Covariant Hamiltonian formalism,
{\sl Mechanics, analysis and geometry: 200 years after Lagrange}.
North-Holland Delta Ser., North-Holland, Amsterdam, 1991, pp. 203-235.

\bibitem{Go3}
M.J. Gotay.
A multisymplectic framework for classical field theory and the calculus of variations, II. Space $+$ time decomposition,
{\sl Differential Geom. App.} {\bf 1} (1991), 375-390.

\bibitem{Gymmsy}
M. J. Gotay, J. Isenberg, J. E. Marsden.
Momentum Maps and Classical Relativistic Fields, Part I: Covariant Field Theory,  arXiv:physics/9801019v2 (2004). Part II:
Canonical analysis of Field Theories, arXiv:math-ph/0411032v1
(2004).

\bibitem{gun}
C. G\"{u}nther.
The polysymplectic Hamiltonian formalism in field theory and calculus of variations I: The local case,
{\sl J. Differential Geom.} {\bf 25} (1987) 23-53.

\bibitem{Kana}
I. V. Kanatchikov.
Canonical structure of classical field theory in the polymomentum phase space,
{\sl  Rep. Math. Phys.} {\bf 41}(1) (1998) 49--90.

\bibitem{Kijo}
J. Kijowski.
A finite-dimensional canonical formalism in the classical field theory,
{\sl Comm. Math. Phys.} {\bf 30} (1973), 99-128.

\bibitem{KijoSz}
J. Kijowski, W. Szczyrba.
Multisymplectic manifolds and the geometrical construction of the Poisson brackets in the classical field theory,
{\sl G\'{e}om\'{e}trie symplectique et physique math\'{e}matique} (Colloq. International C.N.R.S., Aix-en-Provence, 1974) (1974) 347-349.

\bibitem{KijTul}
J. Kijowski, W. M.  Tulczyjew.
 A symplectic framework for field theories.
{\sl Lecture Notes in Physics}, {\bf 107}. Springer-Verlag, New York, 1979.

\bibitem{KMS93} I. Kol\'ar, P.W. Michor, J. Slovak. Natural operations
  in differential geometry, {\sl Springer-Verlag}, 1993.


\bibitem{kosambi48}
D.D. Kosambi. Systems of partial differential equations of the second order,
{\sl Quart. J. Math}., {\bf 19}(1948), 204--219.

\bibitem{LMMS1}
 M. de Le\'{o}n, D. Mart\'{\i}n de Diego, M. Salgado, S. Vilari\~{n}o.
Nonholonomic constraints in k-symplectic Classical Field Theories.
{\sl International Journal of Geometric Methods in Modern Physics}, {\bf 5}(5) (2008) 799-830.

\bibitem{LMMS2}
M. de Le\'{o}n, D. Mart\'{\i}n de Diego, M. Salgado, S. Vilari\~{n}o.
k-symplectic formalism on Lie algebroids.
{\sl J. Phys. A: Math. Theor}. {\bf 42} (2009) 385209 (31 pp)

\bibitem{LM-96}
 M. de Le\'on, D. Mart\'\i n de Diego.
Symmetries and Constant of the Motion for Singular Lagrangian Systems,
 {\sl Int. J. Theor. Phys.} {\bf 35}(5) (1996) 975-1011.

\bibitem{LMS-2004}
 M. de Le\'on, D. Mart\'\i n de Diego, A. Santamar\'\i a-Merino.
Symmetries in classical field theories,
{\sl Int. J. Geom. Meth. Mod. Phys.} {\bf 1}(5) (2004) 651-710.

 \bibitem{mt1}
M. de Le\'{o}n, I. M\'endez, M. Salgado.
$p$-almost tangent structures,
{\sl Rend.   Circ. Mat. Palermo}, Serie II {\bf XXXVII} (1988), 282--294.

\bibitem{mt2}
 M. de Le\'{o}n, I. M\'{e}ndez, M. Salgado.
Integrable $p$--almost tangent structures and tangent bundles of $p^1$-ve\-lo\-ci\-ties,
{\sl Acta Math. Hungar.} {\bf 58}(1-2) (1991), 45--54.

\bibitem{mod1}
M. de Le\'{o}n, E. Merino, J.A. Oubi\~{n}a, P. Rodrigues,  M. Salgado.
Hamiltonian systems on $k$-cosymplectic manifolds,
{\sl J. Math. Phys.} {\bf 39}(2) (1998)  876--893.

\bibitem{mod2}
M. de Le\'{o}n, E. Merino,  M. Salgado.
$k$-cosymplectic manifolds and Lagrangian field theories,
{\sl J. Math. Phys.} {\bf 42}(5) (2001) 2092--2104.

\bibitem{Maranon}
H. Mara\~{n}\'{o}n.
Simetries d'equacions diferencials. Aplicaci\'{o} als sistemes k-simpl\`ectics,
Treball Fi de Master Matematica Aplicada. 2008 Department of Applied Mathematics IV. Technical University of Catalonia (UPC).

\bibitem{marmo86}
G. Marmo, N. Mukunda.
Symmetries and constants of the motion in the Lagrangian formalism on TQ: beyond point transformations,
{\sl Nuovo Cim. B}, \textbf{92} (1986) 1--12.

\bibitem{mrsv}
J.C. Marrero, N. Rom\'an-Roy, M. Salgado, S. Vilarino.
On a kind of Noether symmetries and conservation laws in k-cosymplectic Field Theory,
{\sl Journal of Mathematical Physics} {\bf 52}, 022901 (2011), 20 pp

\bibitem{morimoto}
A. Morimoto.
Liftings of some types of tensor fields and connections to tangent $p^r$-velocities,
{\sl Nagoya Qath. J.} {\bf 40} (1970) 13-31.

\bibitem{fam}
F. Munteanu, A. M. Rey, M. Salgado.
The G\"{u}nther's formalism in classical field theory: momentum map and reduction,
{\sl J. Math. Phys.} {\bf 45}(5) (2004) 1730--1751.

\bibitem{Mart2}
M. C. Mu\~{n}oz-Lecanda, M. Salgado, S. Vilari\~{n}o.
k-symplectic and k-cosymplectic Lagrangian field theories: some interesting examples and applications.
{\sl International Journal of Geometric Methods in Modern Physics}, {\bf 7}(4) (2010) 669-692.

\bibitem{No2}
 L.K. Norris.
 Generalized symplectic geometry on the frame bundle of a manifold,
 {\sl Proc. Symp. Pure Math.}  {\bf 54}, Part 2 (Amer.
Math. Soc., Providence RI, 1993), 435-465.

\bibitem{No3}
L.K. Norris.
Symplectic geometry on $T^*M$ derived from $n$-symplectic geometry on $LM$.
{\sl J. Geom.\ Phys.} {\bf 13}
(1994) 51-78.

\bibitem{No4}
L.K. Norris.
 Schouten-Nijenhuis Brackets,
{\sl J. Math.\ Phys.} {\bf  38} (1997) 2694-2709.

\bibitem{No5}
L. K. Norris.
 $n$-symplectic algebra of observables in covariant Lagrangian field theory,
 {\sl J. Math. Phys.} {\bf
42}(10) (2001) 4827--4845.


\bibitem{Olver}
 P.J.  Olver.
Applications of Lie groups to differential equations,
{\sl Graduate Texts in Mathematics}, {\bf 107}. Springer-Verlag, New York, 1986.

\bibitem{relations}
N. Rom\'{a}n-Roy, A. M. Rey, M. Salgado, S. Vilari\~{n}o.
On the k-Symplectic, k-Cosymplectic and Multisymplectic Formalism of
Classical Field Theories, {\sl Journal of Geometric Mechanics} {\bf 3}(1), March 2011

\bibitem{rsv07}
N. Rom\'an-Roy, M. Salgado,  S. Vilari\~{n}o.
Symmetries and Conservation Laws in G\"unter k-symplectic formalism of Field Theory,
{\sl Reviews in Mathematical Physics}, {\bf 19} (10) (2007), 1117--1147.

\bibitem{Snia}
 J. Sniatycki.
On the geometric structure of classical field theory in Lagrangian formulation,
{\sl Math. Proc. Cambridge Philos. Soc.} {\bf 68} (1970) 475-484.

\bibitem{Tulczy1}
W.M. Tulczyjew.
Hamiltonian systems, Lagrangian systems and the Legendre transformation,
{\sl Symposia Mathematica} {\bf 16} (1974) 247--258.

\bibitem{yano73}
K.Yano, S. Ishihara.
 Tangent and cotangent bundles,
 {\sl Marcel Dekker}, Inc., 1973.

 \end{thebibliography}
\end{document}